\documentclass[preprint,review,1p]{elsarticle}

\usepackage[utf8]{inputenc} 

\usepackage{graphicx}
\usepackage{pgfplots,tikz}
\usepackage{tikz-cd}

\usepackage{amsmath,amssymb,amstext,mathtools,amsthm}

\usepackage{bm}
\newtheorem{theorem}{Theorem}[section]
\newtheorem{lemma}[theorem]{Proposition}
\usepackage{algpseudocode}
\usepackage{algorithm}

\usepackage{caption, subcaption}
\usepackage{multirow}
\usepackage{booktabs}
\usepackage{makecell}
\newcommand{\ra}[1]{\renewcommand{\arraystretch}{#1}}
\usepackage{multirow}

\usepackage{color}
\usepackage{threeparttable}
\newcommand{\blue}[1][\textcolor{black}]{#1}

\usepackage{titlesec}
\usepackage{siunitx}
\usepackage{ifpdf}

\begin{document}
	
\begin{frontmatter}
%
\title{Robust low-rank covariance matrix estimation with a general pattern of missing values}

\author[1,*]{A. Hippert-Ferrer}
\author[2]{M. N. El Korso}
\author[2]{A. Breloy}
\author[3]{G. Ginolhac}
\address[1]{Université Paris-Saclay, CNRS, CentraleSupélec, Laboratoire des signaux et systèmes, 91190, Gif-sur-Yvette, France.}
\address[2]{LEME, Paris-Nanterre University, 50 rue de S\`evres, 92410 Ville d'Avray, France}
\address[3]{LISTIC, Savoie Mont Blanc University, 5 chemin de Bellevue, Annecy, France}
\address[*]{Corresponding author: alexandre.hippert-ferrer@centralesupelec.fr}

\title{Robust low-rank covariance matrix estimation with a general pattern of missing values}

\begin{abstract}
This paper tackles the problem of robust covariance matrix estimation when the data is incomplete. Classical statistical estimation methodologies are usually built upon the Gaussian assumption, whereas existing robust estimation ones assume unstructured signal models. The former can be inaccurate in real-world data sets in which heterogeneity causes heavy-tail distributions, while the latter does not profit from the usual low-rank structure of the signal. Taking advantage of both worlds, a covariance matrix estimation procedure is designed on a robust (\blue{mixture of scaled} Gaussian) low-rank model by leveraging the observed-data likelihood function within an expectation-maximization algorithm. It is also designed to handle general pattern of missing values. The proposed procedure is first validated on simulated data sets. Then, its interest for classification and clustering applications is assessed on two real data sets with missing values, which include multispectral and hyperspectral time series.
\end{abstract}


\end{frontmatter}

%

\section{Introduction}
\label{sec1}

Missing data appear when no value of the data is available for a given variable and a given observation. This classical problem \cite{anderson1957,afifi1966} is a pitfall in statistical signal processing and its related fields, as statistical inference \cite{little1987, schafer1997} and data analysis \cite{benzecri1973,vanbuuren2018}. To name a few applications where missing data has drawn significant attention, we can cite  biomedical studies, chemometrics \cite{walczak1} or remote sensing where missing values created by poor atmospheric conditions or sensor failure can dramatically hamper the understanding of the physical phenomenon under observation \cite{shen2015}. Covariance matrix (CM) estimation theory, which is a fundamental issue in signal processing and machine learning problems, has witnessed particular efforts focusing in the case of incomplete data.

In this scope, it is known that any efficient estimation algorithm should be able to exploit the source of information offered by missing values, which is formerly called \textit{informative missingness} \cite{rubin1976}. One approach to estimate the CM with missing values is to rely on maximum likehood (ML) estimation with a prior assumption on the probability distribution of the data. Within this framework, the Expectation-Maximization (EM) algorithm \cite{Dempster1977} is a handy iterative procedure to obtain ML estimates as it is based on the expectation of the conditional probability $p(\bm{z}|\bm{x},\bm{\theta})$ of the latent (missing) variables $\bm{z}$ given the observed variables $\bm{x}$ and the parameters $\bm{\theta}$ under estimation. Extensive work has been put into CM estimation with missing values using the EM algorithm by assuming independent and identically distributed (iid) samples drawn from the Gaussian distribution \blue{when the sample size $n$ exceeds the dimension $p$ ($n>p$)} \cite{jamshidian1997, liu1999} or in high-dimensional regime \cite{schneider2001,lounici2014,stadler2014} ($p>n$). Other models have been considered in which the covariance is assumed to have a low-rank plus identity structure \cite{lounici2014}, which will be referred in the following to as LR structure\footnote{It is sometimes named spiked or factor model.}. This structure is closely related to signal subspace inference or principal component analysis (PCA) with missing values \cite{chen2009, josse2012}.

To overcome robutsness issues associated with the classical Gaussian assumption, more general distributions have been considered, such as the multivariate $t$-distribution \cite{little1988, liu1995} and its regularized version for small sample size using an improved EM algorithm \cite{liupalomar2019}. These distributions are encompassed by elliptical symmetric (ES) distributions, which are directly linked via ML estimation to the so-called M-estimators in the complete data case \cite{tyler1987, conte2002, pascal2008}. The robustness of these tools to non-Gaussianity \cite{zoubir2018} has been illustrated in a wide range of applications including radar processing, hyperspectral imagery and classification \cite{gini2000, greco2007,theiler2005,formont2010,ollier2018,abdallah2020}. \blue{Interestingly, this family can be accounted for in a robust (distribution-free) manner by considering the so-called mixture of scaled Gaussian (\blue{MSG}) \cite{ollila2012} , which models the data as Gaussian conditionally to an unknown deterministic scale for each sample.}

This paper proposes to take advantage of both robust estimation and LR structure models by proposing a EM-based procedure to estimate the CM in the case of incomplete data drawn from a \blue{MSG} distribution. As cited above, existing robust covariance estimation algorithms based on \blue{MSG} distributions are not designed to deal with missing values, and, to the best of our knowledge, existing works which use the LR structure are only based on Gaussian assumptions. In the full rank case, the work of \cite{frahm2010} has extended generalized elliptical distributions (GES), from which \blue{MSG} distributions are a sub-class, to incomplete data by proposing an adapted form of the Tyler's M-estimator.

Furthermore, existing estimation algorithms that handle missing data assume that the missingness pattern, \emph{i.e.}, the pattern describing which values are missing with respect to the observed data, is a monotone pattern \cite{little1988,liu1999,liupalomar2019,frahm2010} (see the illustration of missing data patterns in Fig.~\ref{fig1}). This pattern can be of interest, \emph{e.g.}, in longitudinal studies \cite{little2019} or sensor failure \cite{larsson2001}. However, in other applications such as remote sensing, missing values can take very diverse patterns because of unpredictable events (clouds, snowfall, etc.) \cite{shen2015}, which leads to the so-called general pattern \cite{little1987}.


\begin{figure}
	\centering
	\includegraphics[scale=1.1, trim={4.1cm 23.2cm 8cm 5.1cm}, clip]{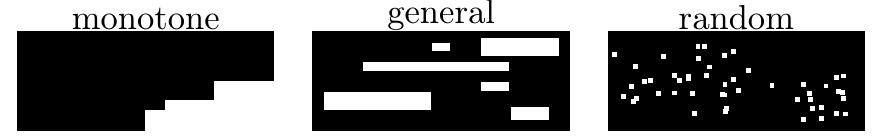}
	\caption{Illustration of three rectangular data sets with different missing data patterns (black=observed, white=missing): monotone, general and random. Note that the random pattern is a special case of the general pattern, where missing values appear on individual observations rather than blocks. See \cite{little2002} (p. 5) for a precise description of missing data patterns.}
	\label{fig1}
\end{figure}

With the previous points in mind concerning \textit{i)} robustness to non-Gaussianity, \textit{ii)} low-rank structured models and \textit{iii)} missingness patterns, the contributions of this paper are summarized below:

\begin{itemize}
	\item[1)] A generic algorithm for CM estimation is developed for structured signals with a non-Gaussian distribution. A procedure for non-structured / Gaussian is also obtained as a special case. In the structured configuration, the covariance matrix is supposed to have a LR structure as in \cite{goodman2007}, which will be detailed in the next section;
	\item[2)] The analysis of missingness patterns is extended to the general pattern, which fills a gap in robust estimation; 
	\item[3)] The proposed estimators are tested on both simulated and real data sets with two applications in machine learning, namely supervised classification and unsupervised clustering with missing values, which has rarely been unfolded (outside of simulated data) in the literature of CM estimation with missing values.
\end{itemize}

The rest of the paper is organized as follows: Section~\ref{sec2} formulates the problem by framing the \blue{MSG} distribution, the LR CM structure and the missing data model. Section~\ref{sec3} describes the proposed robust CM estimation procedures for unstructured and structured models. Section~\ref{sec4} illustrates the performance of the proposed method on simulated data sets in terms of CM estimation and data imputation. Finally, Section~\ref{sec5} shows the interest of the proposed procedure in real data applications using covariance-based machine learning methods.

\noindent \textit{Notations}. $a$ indicates a scalar quantity, $\bm{a}$ represents a vector quantity and $\bm{A}$ a matrix. $\{a_i\}_{i=1}^n$ denotes the set of elements $a_i$ with $i\in[1,n]$. The transpose operator is $\top$, whereas $\text{tr}\{\cdot\}$ and $|\cdot|$ are respectively the trace and the determinant operators. The eigenvalue decomposition \blue{of $\bm{A}$ into matrices $\bm{U}$ and $\bm{\Lambda}$} is denoted by \blue{$ \bm{A} \overset{\text{EVD}}{=} \bm{U}\bm{\Lambda}\bm{U}^\top$}. \blue{The symbol} $\succeq$ indicates positive definitiveness, $\mathcal{S}^p_{++}$ is the set of $p\times p$ symmetric positive definite (SPD) matrices\footnote{$\mathcal{S}_{++}^p = \{\bm{\Sigma}\in \mathcal{S}_p: \forall \bm{x}\in \mathbb{R}^p \backslash \{0\}, \bm{x}^{\top}\bm{\Sigma}\bm{x}>0 \}$.} and $\text{St}_{p,r}$ is the real Stiefel manifold of $p\times r$ orthogonal matrices\footnote{$\text{St}_{p,r} = \{\bm{U} \in \mathbb{R}^{p\times r} : \bm{U}^\top\bm{U}=\bm{I}_r\}$.}. $\propto$ stands for ``proportional to". $Q(\cdot)$ is the score function. Finally, $\mathbb{E}[\cdot]$ denotes the expectation operator.

\section{Problem formulation}
\label{sec2}

\subsection{\blue{Mixture of scaled Gaussian} distributions}

Most of covariance matrix estimation procedures with incomplete data use the Gaussian assumption. However, this assumption can be inaccurate in remote sensing applications, where images often include heterogeneous areas. The \blue{mixture} of scaled Gaussian (\blue{MSG}) distributions can tackle this issue by managing heavier tails, which offers a better fit to empirical data \cite{ollila2012}. A real $p$-vector $\bm{y} \in \mathbb{R}^p$ follows a zero-mean multivariate \blue{MSG} distribution\footnote{\blue{Note that this distribution can be easily transposed to complex-valued vectors \cite{ollila2012}.}} if it admits the stochastic representation

\begin{equation}
	\label{eq:stochastic_y}
	\bm{y} = \sqrt{\tau} \bm{n}
\end{equation}
with $\bm{n} \sim \mathcal{N}(\bm{0},\bm{\Sigma})$ and for some scalar $\tau \in \mathbb{R}_{>0}$, called the \textit{texture}, \blue{which is strictly positive, deterministic and unknown for each sample.} 

Let us now define a rectangular data set $\bm{Y} \in \mathbb{R}^{p\times n}$ represented by $\bm{Y} = \{\bm{y}_i = (y_{1,i}, y_{2,i}, \dots, y_{p,i})^{\top} \}$ where $\{ \bm{y}_i \}_{i=1}^n$ are modeled as $n$ iid vectors of dimension $p$ drawn from a \blue{MSG} distribution. This leads to the following model: 

\begin{equation}
\label{eq:distribution}
\bm{y}_i|\tau_i \sim \mathcal{N}(\bm{0},\tau_i\bm{\Sigma}), \; \; \bm{\Sigma} \subseteq \mathcal{S}^p_{++}, \; \; \tau_i > 0
\end{equation}
The \blue{log}likelihood function of model~(\ref{eq:distribution}) is given by

\begin{equation}
	\label{eq:loglikelihood}
	\blue{\log \ell(\{\bm{y}_i\}|\bm{\Sigma},\{\tau_i\}) } \propto -n\log|\bm{\Sigma}| - p\sum_{i=1}^{n}\log \tau_i - \sum_{i=1}^{n}\bm{y}_i^{\top}(\tau_i\bm{\Sigma})^{-1}\bm{y}_i
\end{equation}
In this model, the texture can be seen as a \textit{scale} setting of the Gaussian model \cite{wiesel2012}. \blue{Such model is also intertwined with the class of compound Gaussian (CG) distributions (which is a subclass of ES distributions), which assumes the texture independent from $\bm{n}$ in (\ref{eq:stochastic_y}) and with a given PDF $f_\tau(.)$}. Note that the type of distribution within the class of CG distributions is mainly guided by the assumption on the PDF function \cite{ollila2012}.

\blue{Considering deterministic $\{\tau_i\}$ instead of assuming a PDF $f_\tau$ has proven its convenience in terms of robustness\footnote{This model follows a parametrization of the covariance matrix of the real elliptical model $\mathcal{E}(\bm{0},\tau \bm{\xi})$, where $\bm{\xi}$ is the shape matrix and where $\tau$ has only one value. In our case, the scale parameter varies at each observation $i=1,\dots,n$. More details can be found in \cite{yao1973,zoubir2018,maronna2019}.} \cite{tyler1987, pascal2008, ollila2012}}. This distribution is also more robust than the purely Gaussian one because the scales allow flexibility in the presence of heterogeneous data, \emph{e.g.}, noisy data, possible outliers or inconsistencies in the data (see, \textit{e.g.}, \cite{mian2018} for Synthetic Aperture Radar data).

\textbf{Remark 1.} As it is clear that $\tau_i=1, \; \forall i$ in the deterministic \blue{MSG} distribution~(\ref{eq:distribution}) gives the Gaussian distribution, results regarding CM estimation will be given by considering the Gaussian distribution as a special case of the \blue{MSG} distribution.

Here we assume that $\bm{\Sigma}$ is characterized by a LR structure which can be modeled by the well-known factor model \cite{ruppert2011} (also known as spiked model \cite{johnstone2001}):
\begin{equation}
\begin{cases}
\bm{\Sigma} = \sigma^2\bm{I}_p + \bm{H} \\
\bm{H} \succeq \bm{0} \\
\text{rank}(\bm{H})=r
\label{eq:lowrank}
\end{cases}
\end{equation}
where $\bm{I}_p$ denotes the $p$-dimensional identity matrix and $\bm{H}$ is a $p\times p$ low-rank signal covariance matrix of rank $r$. As in many works, the latter is considered to be known from prior physical assumptions \cite{goodman2007} or pre-estimated as in model order selection techniques \cite{stoica2004}. Model~(\ref{eq:lowrank}) is directly related to principal component analysis (PCA) and subspace recovery \cite{tipping1999}. 

\subsection{Data model}

As each vector $\bm{y}_i$ might have missing elements, it is necessary to design a model that takes into account incompleteness. Thereby, each of its observed and missing elements can be grouped in vectors denoted by $\bm y_i^o$ and $\bm y_i^m$ respectively, which are stacked in vectors $\widetilde{\bm y}_i$ such that 

\begin{equation}
\label{eq:model_data}
\widetilde{\bm y}_i = \bm P_i \bm y_i = \begin{pmatrix}
\bm{y}_{i}^o \\
\bm{y}_{i}^m
\end{pmatrix}, \quad i=1,\dots,n
\end{equation} 
where $\bm{P}_i \in \mathbb{R}^{p\times p}$ denotes a permutation matrix\footnote{Note that $\bm{P}_i$ is invertible and $\bm{P}_i^{-1} = \bm{P}_i^{\top}$.}. As illustrated by Fig.~\ref{fig2}, this set of operations consists in permuting the missing elements at the bottom of each vector $\bm{y}_i$. Then, the covariance matrix of $\widetilde{\bm y}_i$ becomes

\begin{equation}
	\widetilde{\bm{\Sigma}}_i=\begin{pmatrix}
	\widetilde{\bm{\Sigma}}_{i,oo} & \widetilde{\bm{\Sigma}}_{i,mo} \\
	\widetilde{\bm{\Sigma}}_{i,om} & \widetilde{\bm{\Sigma}}_{i,mm}
	\end{pmatrix}=\bm{P}_i\bm{\Sigma}\bm{P}_i^{\top}
\end{equation}
where $\widetilde{\bm{\Sigma}}_{i,mm}$, $\widetilde{\bm{\Sigma}}_{i,mo}$, $\widetilde{\bm{\Sigma}}_{i,oo}$ are the block CM of $\bm{y}_i^m$, of $\bm{y}_i^m$ and $\bm{y}_i^o$, and of $\bm{y}_i^o$.

In a similar manner, the recent work of \cite{aubry2021} \blue{models (\ref{eq:model_data}) by} $\widetilde{\bm{y}}_i = \bm{A}_i\bm{y}_i$, where $\bm{A}_i$ is a $n_i\times p$ selection matrix \blue{constructed from the extraction of the $n_i \le n$ rows of $\bm{I}_p$ corresponding to the available observations at the $i$-th snapshot}. This is only an alternative view of the problem. However, \blue{the derived} results \blue{of \cite{aubry2021}} hold for the Gaussian distribution. In our case, it is worth mentioning that since each $\bm{y}_i$ follows a \blue{MSG} distribution, their permuted version $\widetilde{\bm{y}}_i$ also follow a \blue{MSG} distribution because only elements within each $p$-vector are permuted.

\tikzset{every picture/.style={line width=0.5pt}} 
\begin{center}
	\begin{figure}
	\centering
	\includegraphics[trim={2cm 23.3cm 11cm 1.7cm}, clip]{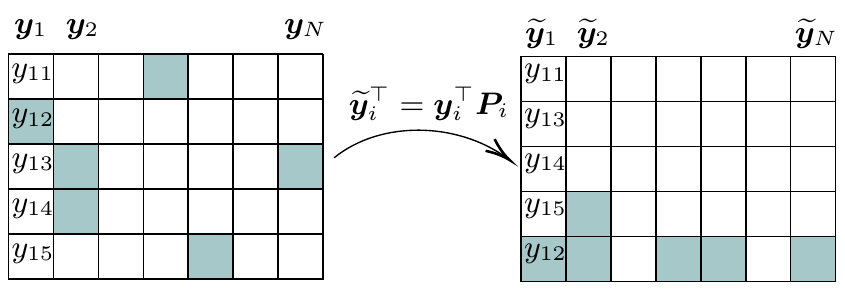}
	\caption{Missing elements (colored) of each $\bm{y}_i$ are permuted to fall at the bottom of each $\widetilde{\bm{y}}_i$. \blue{Notice that} a completely observed $\bm{y}_i$ implies $\bm{P}_i = \bm{I}_p$.}
	\label{fig2}
	\end{figure}
\end{center}

\subsection{The EM algorithm: a brief reminder}

The EM algorithm is a widely employed iterative scheme for ML estimation in incomplete data problems \cite{Dempster1977}. This algorithm offers a rigorous and formal approach to the intuitive \textit{ad hoc} idea of filling in missing values\footnote{Indeed, it is not the missing values themselves that are filled, but the function of the missing values (the sufficient statistics) that are computed \cite{little2002}.}: at the E-step, the following conditional expectation of the complete data loglikelihood is found given the observed data and the current \blue{deterministic} estimated parameters $\bm{\theta}$:

\begin{equation}
\mathcal{L}_c(\blue{\widetilde{\bm{Y}}}|\bm{\theta}) = \log \blue{\ell(\{\bm{y}_i\}}|\bm{\Sigma},\blue{\{\tau_i\}}) 
\end{equation}
where $\bm{\theta}$ is the vector of unknown parameters \blue{(which are $\tau_i$ and $\bm{\Sigma}$ in our case) and $\widetilde{\bm{Y}} = [\widetilde{\bm{y}}_1,\dots,\widetilde{\bm{y}}_n]$. Note that $\mathcal{L}_c(\widetilde{\bm{Y}}|\bm{\theta})$ is equivalent to $\mathcal{L}_c(\bm{Y}|\bm{\theta})$ since both $\widetilde{\bm{Y}}$ and $\bm{Y}$ contain the same missing values, ordered and unordered, respectively.} \blue{$\widetilde{\bm{Y}}$ (or $\bm{Y}$)} are the so-called \textit{complete data} which are the combination of observed and missing data. At the M-step, the parameters are updated by maximizing the expected complete data loglikelihood. To summarize, if $\bm{\theta}^{(t)}$ is the current estimate of the parameter $\bm{\theta}$ and $f(\cdot)$ the density function, the E-step computes 

\begin{equation}
	\label{eq:e-step}
	Q(\bm{\theta}|\bm{\theta}^{(t)}) = \int \mathcal{L}_c(\widetilde{\bm{Y}}|\bm{\theta})f(\bm{y}^m|\bm{y}^o,\bm{\theta}=\bm{\theta}^{(t)})d\bm{y}^m
\end{equation}
and the M-step find $\bm{\theta}^{(t+1)}$ by maximizing (\ref{eq:e-step}):

\begin{equation}
	Q(\bm{\theta}^{(t+1)}|\bm{\theta}^{(t)}) \ge Q(\bm{\theta}|\bm{\theta}^{(t)})
\end{equation}
The goal is then to repeat E and M-steps until a stopping criteria, such as the distance $||\bm{\theta}^{(m+1)}-\bm{\theta}^{(m)}||$, converges to a pre-defined threshold.

\section{Covariance estimation under non-Gaussian environment}
\label{sec3}

In this section, covariance matrix estimation procedures are developed in the presence of missing data under the \blue{MSG} distribution. The rank constraint given by (\ref{eq:lowrank}) is also resolved. For both configurations, the EM algorithm is adopted. In the following, we consider the determinant-based normalization \cite{paindaveine2008} which leads to the estimation of the normalized CM $\bm{\Sigma}/|\bm{\Sigma}|^{\frac{1}{p}}$, called the shape matrix \cite{zhang2017}.

\subsection{Robust full-rank estimation}

In this case, the unknown parameters to estimate are $\bm{\theta}=[\bm{\zeta}^{\top},\{\tau_i\}_{i=1}^n]^{\top}$, where $\bm{\zeta}$ contains the elements of the lower triangular matrix of $\bm{\Sigma}$ including its diagonal \blue{and where $\{\tau_i\}$ are the texture parameters}. As we shall see hereafter, the differences with the purely Gaussian case ($\tau_i=1$) are that 1) the unknown scales $\tau_i$ must be taken into account in the formulation and computation of the expectation at the E-step and 2) a closed-form expression $\widehat{\tau}_i$ must be derived to update $\widehat{\bm{\Sigma}}$ at the M-step. \blue{These estimators can be found by maximization the loglikelihood of the incomplete data, which is formulated onward.} Replacing $\bm{y}_i$ by its permuted version $\widetilde{\bm{y}}_i$ in~(\ref{eq:loglikelihood}), the complete data loglikelihood of the \blue{MSG} distribution is alternatively given by:

\begin{equation}
\label{eq:ll_CD_permuted}	
\mathcal{L}_{\text{c}}(\blue{\widetilde{\bm{Y}}}|\bm{\theta}) \propto - n\log|\bm{\Sigma}| - p\sum_{i=1}^{n}\log \tau_i-\sum_{i=1}^{n}\widetilde{\bm{y}}_{i}^{\top}\big(\tau_i\widetilde{\bm{\Sigma}}_i\big)^{-1}
\widetilde{\bm{y}}_{i}
\end{equation}
At the E-step of the algorithm, the expectation of the complete loglikelihood~(\ref{eq:ll_CD_permuted}) is computed by using the so-called $Q$-function. This function is the expectation of the missing data conditioned by the estimation of the parameters at the $t$-th iteration of the algorithm, that is:

\begin{equation}
\label{eq:Q_function}
Q_i(\bm{\theta}|\bm{\theta}^{(t)}) = \mathbb{E}_{\bm{y}_i^m|\bm{y}_i^o,\bm{\theta}^{(t)}} \big[\mathcal{L}_{\text{c}}(\bm{y}_i^o,\bm{y}_i^m|\bm{\theta})\big]
\end{equation}
Due to the iid of the observations, one obtains:

\begin{equation}
Q(\bm{\theta}|\bm{\theta}^{(t)}) = \sum_{i=1}^{n} Q_i(\bm{\theta}|\bm{\theta}^{(t)})
\end{equation}
where

\begin{align}
\label{eq:Qi}
Q_i(\bm{\theta}|\bm{\theta}^{(t)}) &= \mathbb{E}_{\bm{y}_i^m|\bm{y}_i^o,\bm{\theta}^{(t)}}\Big[
- n\log|\bm{\Sigma}| - p\log \tau_i-\widetilde{\bm{y}}_{i}^{\top}\big(\tau_i\widetilde{\bm{\Sigma}}_i\big)^{-1}
\widetilde{\bm{y}}_{i}\Big] \nonumber \\
&= \mathbb{E}_{\bm{y}_i^m|\bm{y}_i^o,\bm{\theta}^{(t)}}\Bigg[- n\log|\bm{\Sigma}| - p\log \tau_i-\begin{pmatrix}
\bm{y}_i^o \\
\bm{y}_i^m
\end{pmatrix}^{\top} \big(\tau_i\widetilde{\bm{\Sigma}}_i\big)^{-1} \begin{pmatrix}
\bm{y}_i^o \\
\bm{y}_i^m
\end{pmatrix}\Bigg]
\end{align}
The computation of (\ref{eq:Qi}) is hastened as both first and second terms \blue{in the expectation} are deterministic, which means that only the expectation of the last term has to be computed. To calculate \blue{this expectation}, the trace tr$\{\cdot\}$ is used:

\begin{align}
\label{eq:tr}
&\mathbb{E}_{\bm{y}_i^m|\bm{y}_i^o,\bm{\theta}^{(t)}}\Bigg[\begin{pmatrix}
\bm{y}_i^o \nonumber \\
\bm{y}_i^m
\end{pmatrix}^{\top} \big(\tau_i\widetilde{\bm{\Sigma}}_i\big)^{-1} \begin{pmatrix}
\bm{y}_i^o \\
\bm{y}_i^m
\end{pmatrix}\Bigg] \\
&= \mathbb{E}_{\bm{y}_i^m|\bm{y}_i^o,\bm{\theta}^{(t)}} \Bigg[\text{tr}
\bigg\{  \begin{pmatrix}
\bm{y}_{i}^o \\ \nonumber
\bm{y}_{i}^m
\end{pmatrix}
\begin{pmatrix}\bm{y}_{i}^{o\top} & \bm{y}_{i}^{m\top}\end{pmatrix}\big(\tau_i\widetilde{\bm{\Sigma}}_i\big)^{-1} \bigg\} \Bigg] \\
&= \tau_i^{-1} \text{tr}\Bigg\{\mathbb{E}_{\bm{y}_i^m|\bm{y}_i^o,\bm{\theta}^{(t)}} \Bigg[\begin{pmatrix}
\bm{y}_{i}^o\bm{y}_{i}^{o\top} & \bm{y}_{i}^o\bm{y}_{i}^{m\top} \\ \nonumber
\bm{y}_{i}^m\bm{y}_{i}^{o\top} & \bm{y}_{i}^m\bm{y}_{i}^{m\top}
\end{pmatrix}
\Bigg]
\widetilde{\bm{\Sigma}}_i^{-1}\Bigg\} \\
&= \tau_i^{-1} \text{tr} \big\{ \bm{B}_i^{(t)} \widetilde{\bm{\Sigma}}_i^{-1} \big\} 
\end{align}
where $\bm{B}_i^{(t)} = \begin{pmatrix}
\bm{D}_i^{(t)} & \bm{E}_i^{(t)} \\
\bm{F}_i^{(t)} & \bm{G}_i^{(t)}
\end{pmatrix}$ is a $p\times p$ matrix at iteration $t$ of the EM algorithm, with blocks given by

\begin{align}
\label{eq:sufficient_stats}
\bm{D}_i^{(t)} &= \mathbb{E}_{\bm{y}_i^m|\bm{y}_i^o,\bm{\theta}^{(t)}}\big[\bm{y}_{i}^\text{o}\bm{y}_{i}^{o\top}\big] = \bm{y}_{i}^o\bm{y}_{i}^{o\top} \\
\bm{E}_i^{(t)} &=  \blue{\mathbb{E}_{\bm{y}_i^m|\bm{y}_i^o,\bm{\theta}^{(t)}}\big[\bm{y}_{i}^{o}\bm{y}_{i}^{m\top}\big] = \bm{y}_{i}^{o} \mathbb{E}_{\bm{y}_i^m|\bm{y}_i^o,\bm{\theta}^{(t)}}\big[\bm{y}_{i}^{m\top}\big]} \\ 
\bm{F}_i^{(t)} &= \bm{E}_i^{(t)\top}\\ 
\bm{G}_i^{(t)} &= \mathbb{E}_{\bm{y}_i^m|\bm{y}_i^o,\bm{\theta}^{(t)}}\big[\bm{y}_{i}^m\bm{y}_{i}^{m\top}\big]
\end{align}
Thus, the expectation to compute \blue{are $\mathbb{E}_{\bm{y}_i^m|\bm{y}_i^o,\bm{\theta}^{(t)}}[\bm{y}^{m}_i]$ and}  $\mathbb{E}_{\bm{y}_i^m|\bm{y}_i^o,\bm{\theta}^{(t)}}[\bm{y}^m_i\bm{y}^{m^{\top}}_i]$, which are the expectations of the missing data conditioned by the observed data. \blue{Note that $\bm{y}^{m}_i$ and $\bm{y}^m_i\bm{y}^{m^{\top}}_i$ are the sufficient statistics for the complete data $\widetilde{\bm{Y}}$, which have a distribution from the regular exponential family \cite{little2002}}. Here, using a classical result on conditional distributions (see Theorem 2.5.1 in \cite{anderson1965}, p. 35), \blue{the conditional distribution of $\bm{y}_i^m$ given $\bm{y}^o_i$ is}

\begin{equation}
	\blue{\bm{y}_i^m | \bm{y}^o_i \sim \mathcal{N}(\widetilde{\bm{\mu}}_{i,m|o}, \widetilde{\bm{\Sigma}}_{i,mm|o})}
\end{equation}
where $ \widetilde{\bm{\mu}}_{i,m|o} = \widetilde{\bm{\Sigma}}_{i,mo}\widetilde{\bm{\Sigma}}_{i,oo}^{-1}\widetilde{\bm{y}}_i^o$ and $\widetilde{\bm{\Sigma}}_{i,mm|o} = \tau_i(\widetilde{\bm{\Sigma}}_{i,mm} - \widetilde{\bm{\Sigma}}_{i,mo}\widetilde{\bm{\Sigma}}_{i,oo}^{-1}\widetilde{\bm{\Sigma}}_{i,om})$ are the conditional mean and covariance matrix, respectively. \blue{Note that $\tau_i$ is absent from $\widetilde{\bm{\mu}}_{i,m|o}^{(t)}$ as it is annihilated by its inverse.} Following \cite{little2002}, at the $t$-th iteration, the E-step of the EM algorithm consists in calculating
\begin{align}
\label{eq:cond_exp_gaussian}
\blue{\mathbb{E}_{\bm{y}_i^m|\bm{y}_i^o,\bm{\theta}^{(t)}}\big[\bm{y}^m_i\big]} &\blue{= \widetilde{\bm{\mu}}_{i,m|o}^{(t)}} \\ 
\mathbb{E}_{\bm{y}_i^m|\bm{y}_i^o,\bm{\theta}^{(t)}}\big[\bm{y}^m_i\bm{y}^{m^{\top}}_i\big] &= \widetilde{\bm{\Sigma}}_{i,mm|o}^{(t)} + \widetilde{\bm{\mu}}_{i,m|o}^{(t)}\widetilde{\bm{\mu}}_{i,m|o}^{\top(t)}
\end{align}
Finally, $\bm{\theta}^{(t+1)}$ is obtained at the M-step of the algorithm as the solution of the following maximization problem:


\begin{equation}
\label{eq:argmax_CD}
\begin{aligned}
	\max_{\bm{\theta}} & \quad Q_i(\bm{\theta}|\bm{\theta}^{(t)}) \\
	\textrm{subject to} & \quad 
	\begin{array}{lll}
	\bm{\Sigma} \succeq \bm{0} \\
	\tau_i > 0, \; \; \forall i
	\end{array}
\end{aligned}
\end{equation}

\begin{lemma}
	\label{prop:1}
	The ML estimates $\widehat{\bm{\Sigma}}$ and $\widehat{\tau}_i$ of problem (\ref{eq:argmax_CD}) are given by the following closed-form expressions:
	
	\begin{align}
	\widehat{\tau}_i &= \frac{\textnormal{tr}\big\{\bm{B}_i^{(t)}\blue{\bar{\bm{\Sigma}}^{-1}_i}\big\}}{p} \quad \text{for} \quad i \in [1,n] \\
	\widehat{\bm{\Sigma}} &= \frac{p}{n} \sum_{i=1}^{n}\frac{\bm{C}_i^{(t)^{\top}}}{\textnormal{tr}\big\{\bm{C}_i^{(t)}\widehat{\bm{\Sigma}}^{-1}\big\}} \overset{\Delta}{=} \mathcal{H}(\widehat{\bm{\Sigma}})
	\end{align}
	\blue{with $\bar{\bm{\Sigma}}_i = \bm{P}_i \widehat{\bm{\Sigma}} \bm{P}_i^\top$},  $\bm{C}_i^{(t)} = \bm{P}_i^{\top}\bm{B}_i^{(t)}\bm{P}_i$ and where $\mathcal{H}(\cdot)$ is the fixed point equation.
\end{lemma}
\begin{proof}
	
	\blue{Let us start by writing the expectation of the loglikelihood $\mathcal{L}_c(\widetilde{\bm{Y}}|\bm{\theta})$ given by (\ref{eq:ll_CD_permuted}):}
	
	\begin{equation}
	\label{eq:ll_to_differentiate}
	\blue{\mathbb{E}_{\bm{y}_i^m|\bm{y}_i^o,\bm{\theta}^{(t)}}\big[\mathcal{L}_{\text{c}}(\widetilde{\bm{Y}}|\bm{\theta})\big] \propto \mathbb{E}_{\bm{y}_i^m|\bm{y}_i^o,\bm{\theta}^{(t)}}\bigg[- n\log|\bm{\Sigma}| - p\sum_{i=1}^{n}\log \tau_i-\sum_{i=1}^{n}\widetilde{\bm{y}}_{i}^{\top}\big(\tau_i\widetilde{\bm{\Sigma}}_i\big)^{-1}
	\widetilde{\bm{y}}_{i}\bigg]}
	\end{equation}
	\blue{Notwithstanding equation~(\ref{eq:ll_to_differentiate}) is similar to (\ref{eq:Qi}), it differs as it is the expectation conditioned on all observations contained in $\widetilde{\bm{Y}}$, whereas the latter stands for only one observation. By using the result (\ref{eq:tr}) obtained from (\ref{eq:Qi}), it follows that:}
	
	\begin{equation}
		\label{eq:ll_to_differentiate_final}
		\blue{\mathbb{E}_{\bm{y}_i^m|\bm{y}_i^o,\bm{\theta}^{(t)}}\big[\mathcal{L}_{\text{c}}(\widetilde{\bm{Y}}|\bm{\theta})\big] \propto - n\log|\bm{\Sigma}| - p\sum_{i=1}^{n}\log \tau_i-\sum_{i=1}^{n}\tau_i^{-1}\text{tr}\big\{\bm{B}_i^{(t)}\widetilde{\bm{\Sigma}}_i^{-1}\big\}}
	\end{equation}
	\blue{Again, the block matrix $\bm{B}_i^{(t)}$ contains the expectation of the sufficient statistics $\bm{y}^m_i\bm{y}_i^{m\top}$ to be computed at the E-step using (\ref{eq:sufficient_stats}). Now, let differentiate (\ref{eq:ll_to_differentiate_final}) with respect to $\tau_i$ and resolve the equality:}
	\[
	\frac{\partial \blue{\mathbb{E}_{\bm{y}_i^m|\bm{y}_i^o,\bm{\theta}^{(t)}}\big[\mathcal{L}_{\text{c}}(\widetilde{\bm{Y}}|\bm{\theta})\big]}}{\partial \tau_i} = \blue{\bm{0}}
	\]
	This calculus is trivial and finally gives 
	\[\widehat{\tau}_i = \frac{\textnormal{tr}\big\{\bm{B}_i^{(t)}\bar{\bm{\Sigma}}^{-1}_i\big\}}{p}
	\]
	\blue{with $\bar{\bm{\Sigma}}_i= \bm{P}_i\widehat{\bm{\Sigma}}\bm{P}_i^\top$,} which is the desired expression for $\widehat{\tau}_i$. \blue{To find $\widehat{\bm{\Sigma}}$, we first replace} $\tau_i$ by $\widehat{\tau}_i$ in \blue{(\ref{eq:ll_to_differentiate_final})} and obtain the following loglikelihood:
	
	\begin{equation}
	\label{eq:LL2}
	\mathbb{E}_{\bm{y}_i^m|\bm{y}_i^o,\bm{\theta}^{(t)}}\big[\mathcal{L}_{\text{c}}(\widetilde{\bm{Y}}|\bm{\theta})\big] \propto -n\log|\bm{\Sigma}| - p\sum_{i=1}^{n}\log\frac{\textnormal{tr}\big\{\bm{B}_i^{(t)}\widetilde{\bm{\Sigma}}_i^{-1}\big\}}{p} - np 
	\end{equation}
	Then, we derive~(\ref{eq:LL2}) w.r.t. $\bm{\Sigma}$, which leads to the resolution of the following equation:
	
	\[\frac{\partial \blue{\mathbb{E}_{\bm{y}_i^m|\bm{y}_i^o,\bm{\theta}^{(t)}}\big[\mathcal{L}_{\text{c}}(\widetilde{\bm{Y}}|\bm{\theta})\big]}}{\partial \bm{\Sigma}} = \bm{0}
	\]
	\blue{By using the definition $\widetilde{\bm{\Sigma}}_i = \bm{P}_i\bm{\Sigma}\bm{P}_i^\top$, one obtains first:}
	\[
	-n\bm{\Sigma}^{-1} - p\sum_{i=1}^{n}\frac{\bm{\Sigma}^{-1}\bm{C}_i^{(t)^{\top}}\bm{\Sigma}^{-1}}{\textnormal{tr}\big\{\bm{C}_i^{(t)}\bm{\Sigma}^{-1}\big\}} = \bm{0}
	\]
	where $\bm{C}_i^{(t)} = \bm{P}_i^{\top}\bm{B}_i^{(t)}\bm{P}_i$. Arranging terms, then multiplying right and left terms by $\bm{\Sigma}$ gives the desired closed-form expression \blue{of $\widehat{\bm{\Sigma}}$ and concludes the proof}.
\end{proof}

\textbf{Remark 2.}
For $n > p$, which is our case, this estimator can be computed using the fixed point algorithm $\bm{\Sigma}_{m+1}=\mathcal{H}(\bm{\Sigma}_{m})$ where $m$ refers to the iteration index \blue{of the fixed point.}

\textbf{Remark 3.}
In the Gaussian case, finding $\bm{\Sigma}$ that maximizes this expression requires to derive $Q$ with respect to $\bm{\Sigma}$ and then to solve $\frac{\partial Q}{\partial \bm{\Sigma}}=\frac{\partial \mathbb{E}_{\bm{y}_i^m|\bm{y}_i^o,\bm{\Sigma}^{(t)}}\big[\mathcal{L}_{\text{c}}(\widetilde{\bm{Y}}|\bm{\Sigma})\big]}{\partial \bm{\Sigma}}=0$ \blue{where $ \mathbb{E}_{\bm{y}_i^m|\bm{y}_i^o,\bm{\Sigma}^{(t)}}\big[\mathcal{L}_{\text{c}}(\widetilde{\bm{Y}}|\bm{\Sigma})\big] \propto - n\log|\bm{\Sigma}| -\sum_{i=1}^{n}\text{tr}\big\{\bm{B}_i^{(t)}\widetilde{\bm{\Sigma}}_i^{-1}\big\}$}. Based on the above and after some basic derivation calculus, the solution to be computed at each iteration of the EM algorithm is given by:

\begin{equation}
\label{eq:gaussian_sigma}
\widehat{\bm{\Sigma}} = \frac{1}{n}\bm{C}^{(t)^{\top}}
\end{equation}
with $\bm{C}^{(t)} = \sum_{i=1}^n \bm{P}_i^{\top}\bm{B}_i^{(t)}\bm{P}_i$. One can notice that in the case of no missing values, this result leads to the classical Sample Covariance Matrix (SCM) with $\bm{P}_i = \bm{I}_p$ and $\blue{\bm{B}_i^{(t)}} = \mathbb{E}[\bm{y}_i\bm{y}_i^{\top}]$.

The complete estimation procedure of $\bm{\theta}$ is given in Algorithm~\ref{alg1:em_tyler}. The stopping condition of the EM algorithm is ensured by the evaluation of the quantity $||\bm{\theta}^{(t+1)} - \bm{\theta}^{(t)}||^2_F$ at each iteration, while the convergence of the fixed point algorithm relies on $||\bm{\Sigma}_{m+1}^{(t)} - \bm{\Sigma}_{m}^{(t)}||^2_F$. 

At the step $t=0$, the estimate $\bm{\Sigma}^{(0)}$ is initialized with Tyler's estimator from available observations in their full dimension $p$, denoted $\widehat{\bm{\Sigma}}_{\text{Tyl-obs}}$. Unlike $\bm{\Sigma}$, incomplete observations make the direct estimation of $\bm{\tau}$ with the fixed point impracticable: thus, as an initialization, all $\tau_i^{(0)}$ are set to one.

\begin{algorithm}
\caption{EM-Tyl: Estimation of $\bm{\theta}$ under \blue{MSG} distribution with missing values.}
\begin{algorithmic}[1]
	\Require $\{\widetilde{\bm{y}_i}\}_{i=1}^n \sim \mathcal{N}(\bm{0}, \tau_i\bm{\Sigma}), \{\bm{P}_i\}_{i=1}^n$
	\Ensure $\widehat{\bm{\Sigma}},\{\widehat{\tau}_i\}_{i=1}^n$
	\State Initialization: \begin{align*}
	\bm{\Sigma}^{(0)} &= \widehat{\bm{\Sigma}}_{\text{Tyl-obs}} \\
	\bm{\tau}^{(0)} &=\bm{1}_N^{\top}
	\end{align*}
	
	\Repeat 
	\Comment{\textcolor{blue}{EM loop, $t$ varies}}
	\State \noindent Compute 
	\begin{align*}
	\bm{E}_i^{(t)} &=  \widetilde{\bm{y}}_i^o\widetilde{\bm{\Sigma}}_{i,mo}^{(t)}\widetilde{\bm{\Sigma}}_{i,oo}^{-1(t)}\widetilde{\bm{y}}_i^o \\
	\bm{G}_i^{(t)} &= \tau_i^{(t)}\big(\widetilde{\bm{\Sigma}}_{i,mm}^{(t)} - \widetilde{\bm{\Sigma}}_{i,mo}^{(t)}\widetilde{\bm{\Sigma}}_{i,oo}^{-1(t)}\widetilde{\bm{\Sigma}}_{i,om}^{(t)}\big)
	+ \widetilde{\bm{\Sigma}}_{i,mo}^{(t)}\widetilde{\bm{\Sigma}}_{i,oo}^{-1(t)}\widetilde{\bm{y}}_i^o\big(\widetilde{\bm{\Sigma}}_{i,mo}^{(t)}\widetilde{\bm{\Sigma}}_{i,oo}^{-1(t)}\widetilde{\bm{y}}_i^o\big)^{\top}
	\end{align*}
	\State Compute 	$\bm{B}_i^{(t)} = \begin{pmatrix}
		\bm{y}_{i}^o\bm{y}_i^{o\top} & \bm{E}_i^{(t)} \\
		\bm{E}_i^{(t)\top}  & \bm{G}_i^{(t)}
	\end{pmatrix}$ 
	\State Compute $\bm{C}_i^{(t)} = \bm{P}_i^{\top}\bm{B}_i^{(t)}\bm{P}_i$
	\Repeat \Comment{\textcolor{blue}{fixed point, $m$ varies (optional loop)}}
	\State $\widehat{\bm{\Sigma}}^{(t)}_{m+1} = \mathcal{H}(\widehat{\bm{\Sigma}}^{(t)}_m)$ 
	\Until $||\bm{\Sigma}^{(t)}_{m+1} - \bm{\Sigma}^{(t)}_m||^2_F$ converges
	\State Compute $\widehat{\tau}_i^{(t)}, \quad i=1,\dots,n$
	\State $t \leftarrow t+1$
	\Until $||\bm{\theta}^{(t+1)} - \bm{\theta}^{(t)}||^2_F$ converges
	
\end{algorithmic}
\label{alg1:em_tyler}
\end{algorithm}

%

\subsection{Robust low-rank estimation}

Let us now consider the case of data whose covariance matrix lives in a subspace of dimension $r<p$. As we recall, this configuration is useful in many real signal applications and is closely related to principal component analysis (PCA) and signal subspace estimation \cite{tipping1999}.

Following model~(\ref{eq:lowrank}), the parameters to estimate are now given by $\bm{\theta}=[\bm{\zeta}^{\top},\sigma^2,\{\tau_i\}_{i=1}^n]^{\top}$, where $\bm{\zeta}$ contains the elements of the lower triangular matrix of $\bm{H}$. The maximization problem~(\ref{eq:argmax_CD}) to find $\bm{\theta}^{(t+1)}$ at the M-step of the EM algorithm becomes the following low-rank estimation problem:



\begin{equation}
\label{eq:LR_maximization}
\begin{aligned}
\max_{\bm{\theta}} & \quad Q_i(\bm{\theta}|\bm{\theta}^{(t)}) \\
\textrm{subject to} & \quad 
\begin{array}{lll}
\bm{\Sigma} = \sigma^2\bm{I}_p + \bm{H} \\
\text{rank}(\bm{\Sigma}) = r \\
\sigma > 0, \; \; \tau_i > 0, \; \; \forall i
\end{array}
\end{aligned}
\end{equation}
A general solution to this problem was found in the seminal work of \cite{tipping1999}, which is recalled subsequently using the formulations of \cite{kang2014} and \cite{sun2016}. \blue{Firstly, let $\bm{\Sigma}^{(t)} \overset{\text{EVD}}{=}
\sum_{i=1}^{p}\lambda_i^{(t)}\bm{u}_i^{(t)}\bm{u}_i^{\top(t)}$ be the eigenvalue decomposition of $\bm{\Sigma}^{(t)}$ at the $t$-th iteration of the EM algorithm, where $\lambda_1^{(t)}<\dots<\lambda_p^{(t)}$ are the eigenvalues of $\bm{\Sigma}^{(t)}$ and $\bm{u}_1^{(t)}, \dots, \bm{u}_p^{(t)}$ the corresponding eigenvectors. Then, the solution of (\ref{eq:LR_maximization}) is given by:}

\begin{equation}
\label{eq:low_rank_sigma}
\widehat{\bm{\Sigma}} = \hat{\sigma}^2\bm{I}_p + \sum_{i=1}^{r}\hat{\lambda}_i\bm{u}_i^{(t)}\bm{u}_i^{\top(t)}
\end{equation}
where

\begin{equation}
\blue{\hat{\sigma}^2} = \frac{1}{p-r}\sum_{i=r+1}^{p}\lambda_i^{(t)}
\label{eq:low_rank_solution1}
\end{equation}
\begin{equation}
\hat{\lambda}_i = \lambda_i^{(t)} - \blue{\hat{\sigma}^2}, \quad \;  i=1,\dots,r
\label{eq:low_rank_solution2}
\end{equation}
The procedure to estimate $\bm{\theta}$ under the LR assumption is given in Algorithm~\ref{alg2:em_tyler_lr}. It uses the same form than Algorithm~\ref{alg1:em_tyler} where equations~(\ref{eq:low_rank_sigma})--(\ref{eq:low_rank_solution2}) are applied just after the fixed point algorithm to each newly estimated $\bm{\Sigma}$. As stated in \cite{sun2016}, the proof of \blue{monotonicity} of the low-rank estimation algorithm is ensured by standards convergence results of the majorization-minimization (MM) algorithm \cite{razaviyayn2013}. Finally, the estimation of the scales $\bm{\tau}$ does not change from the former full-rank algorithm. 

\begin{algorithm}
	\caption{EM-Tyl-r: low-rank estimation of $\bm{\theta}$ under \blue{MSG} distribution with missing values.}
	\begin{algorithmic}[1]
		\Require $\{\widetilde{\bm{y}_i}\}_{i=1}^n \sim \mathcal{N}(\bm{0}, \tau_i\bm{\Sigma}), \{\bm{P}_i\}_{i=1}^n, \text{rank} \; r < p$
		\Ensure $\widehat{\bm{\Sigma}},\{\widehat{\tau}_i\}_{i=1}^n $
		\State Initialize $\bm{\Sigma}^{(0)}, 
		\bm{\tau}^{(0)} $ as in Algorithm~\ref{alg1:em_tyler}.
		
		\Repeat 
		\Comment{\textcolor{blue}{EM loop, $t$ varies}}
		\State \noindent Compute $\bm{E}_i^{(t)}$, $\bm{G}_i^{(t)}$, $\bm{B}_i^{(t)}$ and $\bm{C}_i^{(t)}$ as in Algorithm~\ref{alg1:em_tyler}.
		\Repeat \Comment{\textcolor{blue}{fixed point, $m$ varies (optional loop)}}
		\State $\widehat{\bm{\Sigma}}^{(t)}_{m+1} = \mathcal{H}(\widehat{\bm{\Sigma}}^{(t)}_m)$ 
		\State $\widehat{\bm{\Sigma}}^{(t)}_{m+1} \overset{\text{EVD}}{=} \sum_{i=1}^{p}\lambda_i^{(t)}\bm{u}_i^{(t)}\bm{u}_i^{\top(t)}$
		\State Update $\widehat{\bm{\Sigma}}^{(t)}_{m+1}$ by computing (\ref{eq:low_rank_sigma}), (\ref{eq:low_rank_solution1}), (\ref{eq:low_rank_solution2})
		\State $\widehat{\bm{\Sigma}}^{(t)}_{m+1} = \widehat{\bm{\Sigma}}^{(t)}_{m+1} / \text{tr}\big(\widehat{\bm{\Sigma}}^{(t)}_{m+1}\big)$
		\Until $||\bm{\Sigma}_{m+1}^{(t)} - \bm{\Sigma}_{m}^{(t)}||^2_F$ converges
		\State Compute $\widehat{\tau}_i^{(t)}, \quad i=1,\dots,n$
		\State $t \leftarrow t+1$
		\Until $||\bm{\theta}^{(t+1)} - \bm{\theta}^{(t)}||^2_F$ converges
		
	\end{algorithmic}
	\label{alg2:em_tyler_lr}
\end{algorithm}

%
%

\subsection{Implementation details}

\blue{In both Algorithm~\ref{alg1:em_tyler} and Algorithm~\ref{alg2:em_tyler_lr}, the set of permutations matrices $\{\bm{P}_i\}_{i=1}^n$ is needed as an input to compute the transformed covariance matrices $\widetilde{\bm{\Sigma}}_i = \bm{P}_i \bm{\Sigma} \bm{P}_i^\top$. Indeed, if a vector $\bm{y}_i$ is fully observed, one finds that $\bm{P}_i = \bm{I}_p$, which makes the computation of $\widetilde{\bm{\Sigma}}$ pointless. To avoid these extra computations, one can decompose the loglikelihood~(\ref{eq:ll_CD_permuted}) in the following way:}

\begin{equation}
	\label{eq:ll_obs_mis}	
	\blue{\mathcal{L}_{\text{c}}(\widetilde{\bm{Y}}|\bm{\theta}) \propto - n\log|\bm{\Sigma}| - p\sum_{i=1}^{n}\log \tau_i-\sum_{i\in N_o}\bm{y}_{i}^{\top}\big(\tau_i\bm{\Sigma}_i\big)^{-1} \bm{y}_{i} -\sum_{i\in N_m}\widetilde{\bm{y}}_{i}^{\top}\big(\tau_i\widetilde{\bm{\Sigma}}_i\big)^{-1} \widetilde{\bm{y}}_{i}}
\end{equation}
\blue{where $N_o = \{i_l\}_{l=1}^{n_o}$ and $N_m = \{i_l\}_{l=1}^{n_m}$ are the sets of indices corresponding to fully observed $\{\bm{y}_i\}$ and partially observed $\{\widetilde{\bm{y}}_i\}$, respectively, with $n = n_o + n_m$. Thenceforth, only the set of permutations matrices $\{ \bm{P}_{i_l} \}_{l=1}^{n_m}$ need to be computed. An important consequence is that if $\bm{y}_i$ is fully observed, $\bm{B}_i^{(t)} = \bm{y}_i^o\bm{y}_i^{o\top} = \bm{y}_i\bm{y}_i^\top$. $\bm{C}_i^{(t)}$ becomes:}

\begin{equation}
	\blue{\bm{C}_i^{(t)} =}
	\begin{cases}
	\blue{\bm{I}_p^{\top}\bm{B}_i^{(t)}\bm{I}_p = \bm{y}_i\bm{y}_i^\top \;, \quad i\in N_o}\\
	\blue{\bm{P}_i^{\top}\bm{B}_i^{(t)}\bm{P}_i \;, \quad i\in N_m}
	\end{cases}
\end{equation}
\blue{Empirical experiments have shown that such care in the computation of $\bm{C}_i^{(t)}$ depending on $i\in N_o$ or $i\in N_m$ accelerates the running time of the EM algorithm.} 

Furthermore, the fixed point loop is optional as it can be seen as an inner EM where $\{\tau_i\}$ is the set of latent variables (hence a single update still increases the likelihood). Our empirical experiments evidenced that both versions of the algorithm achieve similar performance, while performing only a single fixed point iteration tends to achieve a faster convergence. 

\section{Numerical simulations}
\label{sec4}

This section illustrates the validation of the proposed algorithms with numerical experiments on simulated data drawn from the multivariate Gaussian and \blue{mixture of} scaled Gaussian distributions. The performances of the proposed covariance estimation procedure are evaluated in regard to three aspects corresponding to different experiments:

\begin{itemize}
	\item[1)] The missing data ratio and patterns in subsection~\ref{subsec:num1};
	\item[2)] The quantity of outliers corrupting the data in subsection~\ref{subsec:num2};
	\item[3)] The possibility to perform CM-based data imputation in the presence of outliers and missing data in subsection~\ref{subsec:num3}.
\end{itemize}

For the sake of the experiment, the CM and scales parameters are known. The CM $\bm{R}$ is chosen to be Toeplitz\footnote{Note that this information on the structure is not used in the estimation procedure.}, which has the form: 
\begin{equation}
	(\bm{R})_{ij} = \rho^{|i-j|}
\end{equation}
Parameter $\rho$, which controls the correlation structure of the CM, is set to $0.7$. Scales $\{\tau_i\}_{i=1}^n$ are drawn from a Gamma distribution with shape parameter $\alpha$ and scale parameter $\beta=\frac{1}{\alpha}$. In all experiments, we fix $\alpha=1$. To generate a covariance matrix admitting the structure (\ref{eq:lowrank}), we compute:

\begin{equation}
	\bm{\Sigma} = \bm{I}_p + \sigma^2\bm{U}\bm{U}^{\top}
\end{equation}
where $\bm{U} \in \text{St}_{p,r}$ is the underlying signal subspace basis obtained from the eigenvalue decomposition of $\bm{R}$ and $\sigma$ is a free parameter corresponding to the signal to noise ratio. 

In our experiments, the data dimension is $p=15$ and $n=\{63,109,190,331,575,1000\}$. Sets $\{ \bm{y}_i \}_{i=1}^n$ are drawn from the \blue{MSG} distribution with covariance $\bm{\Sigma}$. As the aim is to estimate the structured covariance matrix $\bm{\Sigma}$, the estimation is performed on 500 sets $\{ \bm{y}_i \}_{i=1}^n$ simulated for each value of $n$.
Indeed, for the sake of the experiment, missing data are also simulated, which allows a full control on their ratio and pattern. Importantly, as $n$ increases from 63 to 1000, the missing data ratio decreases as $\{44,22,11,5,2,1\}$\%.

The following covariance estimators, which are reported in Table~\ref{tab:estimators}, are considered for comparison:

\begin{itemize}
	\item[\textit{i)}] The covariance matrix for \blue{MSG} distributions obtained from Algorithm~\ref{alg1:em_tyler} (full-rank) and from Algorithm~\ref{alg2:em_tyler_lr} (low-rank), named $\widehat{\bm{\Sigma}}_{\text{EM-Tyl}}$ and $\widehat{\bm{\Sigma}}_{\text{EM-Tyl-r}}$, respectively. 
	\item[\textit{ii)}] The covariance matrix for the Gaussian distribution obtained from Algorithm~\ref{alg1:em_tyler} (full-rank) and from Algorithm~\ref{alg2:em_tyler_lr} (low-rank) with the solution given by (\ref{eq:gaussian_sigma}), named $\widehat{\bm{\Sigma}}_{\text{EM-SCM}}$ and $\widehat{\bm{\Sigma}}_{\text{EM-SCM-r}}$, respectively.
	\item[\textit{iii)}] The \textit{sample covariance matrix} (SCM) from the clairvoyant data (without missing data):
	\begin{equation}
		\widehat{\bm{\Sigma}}_{\text{SCM-clair}}=\frac{1}{n}\sum_{i=1}^{n}\bm{y}_i\bm{y}_i^{\top}
	\end{equation}
	\blue{and its low-rank version, named $\widehat{\bm{\Sigma}}_{\text{SCM-clair-r}}$}.
	\item[\textit{iv)}] The SCM $\widehat{\bm{\Sigma}}_{\text{SCM-obs}}$ estimated from \blue{the set of fully} observed vectors \blue{$\{\bm{y}_{i_l}\}_{l=1}^{n_o} \in N_o $ where $N_o = \{i_l\}_{l=1}^{n_o}$.}
	\item[\textit{v)}] The \textit{fixed point} estimator or Tyler's estimator from the clairvoyant data:
	\begin{equation}
	\widehat{\bm{\Sigma}}_{\text{Tyl-clair}} = \frac{p}{n}\sum_{i=1}^{n} \frac{\bm{y}_i\bm{y}_i^{\top}}{\bm{y}_i\bm{\Sigma}^{-1}\bm{y}_i^{\top}}
	\end{equation}
	\blue{and its low-rank version, named $\widehat{\bm{\Sigma}}_{\text{Tyl-clair-r}}$}.
	\item[\textit{vi)}] Tyler's estimator $\widehat{\bm{\Sigma}}_{\text{Tyl-obs}}$ from \blue{the set of fully} observed vectors \blue{$\{\bm{y}_{i_l}\}_{l=1}^{n_o} \in N_o $ with $N_o = \{i_l\}_{l=1}^{n_o}$.}
	\item[\textit{vii)}] Multiple imputation (MI) \cite{royston2004, vanbuuren2018} from which a robust version is proposed (RMI): \blue{for each $\bm{y}_i$ with missing values $\bm{y}_i^m$}, $q$ \blue{vectors with imputed missing entries (referred to as \textit{imputed vectors} hereafter) $(\widehat{\bm{y}}_{i1},\widehat{\bm{y}}_{i2},\dots,\widehat{\bm{y}}_{iq})$} are generated \blue{with missing values drawn} from a \blue{MSG} distribution: 
	\begin{equation}
		\widehat{\bm{y}}_{ij}^{m} \sim \mathcal{N}(\bm{\mu}_{ij}^o,\sqrt{\tau_{ij}}\bm{\sigma}_{ij}^{o}), \quad j=1,\dots,q
	\end{equation}
	where $\bm{\mu}_{ij}^{o}$ and $\bm{\sigma}_{ij}^{o}$ are the mean and variance of the observed components of $\bm{y}_{ij}$, and $\tau_{ij}$ are the scales parameters drawn from the Gamma distribution with shape parameter $\alpha=1$. The estimated covariance is the mean of the $q$ Tyler's estimators from the $q$ imputed vectors:
	\begin{equation}
		\widehat{\bm{\Sigma}}_{\text{RMI}} = \frac{1}{q}\sum_{j=1}^{q}\widehat{\bm{\Sigma}}_{\text{Tyl},j} =\frac{p}{nq}\sum_{j=1}^{q}\sum_{i=1}^{n} \frac{\widehat{\bm{y}}_{ij}\widehat{\bm{y}}_{ij}^\top}{\widehat{\bm{y}}_{ij}\widehat{\bm{\Sigma}}^{-1}_{\text{Tyl},j}\widehat{\bm{y}}_{ij}^\top}
		\label{eq:rmi}
	\end{equation}
	\blue{The low-rank version is also estimated, named $\widehat{\bm{\Sigma}}_{\text{RMI-r}}$.}
	\item[\textit{viii)}] Robust stochastic imputation (RSI): this procedure is \blue{tantamount to} MI, but with $q=1$.
	\item[\textit{viiii)}] Mean imputation \cite{vanbuuren2018}: missing components of each vector $\bm{y}_i$ are imputed by the mean of the observed components:
	\begin{equation}
		\widehat{\bm{y}}_i^{m} = \bm{\mu}_i^o
	\end{equation} Then, the covariance $\widehat{\bm{\Sigma}}_{\text{Mean-Tyl}}$ is estimated using Tyler's estimator. \blue{The low-rank version is also estimated, named $\widehat{\bm{\Sigma}}_{\text{Mean-Tyl-r}}$.}
\end{itemize}

\begin{table*}
\centering
\ra{1.15}
\blue{
\begin{tabular}{@{}cllc@{}}\toprule
	\textbf{Type} & \textbf{Estimator} & \textbf{Description} & \textbf{Low-rank} \\
	\toprule
	\multirow{2}{*}{\rotatebox[origin=c]{90}{\textsc{em}}} & EM-Tyl & Robust covariance estimation (this study) & yes \\
	& EM-SCM & Covariance estimation (this study) & yes \\
	\hline
	\multirow{4}{*}{\rotatebox[origin=c]{90}{\textsc{no missing data}}} & Tyl-clair & Tyler's estimator on full data \cite{tyler1987} & yes \\
	& Tyl-obs & Tyler's estimator on fully observed data \cite{tyler1987} & no \\
	& SCM-clair & Sample covariance matrix on full data & yes \\
	& SCM-obs & Sample covariance matrix on fully observed data & no \\
	\hline
	\multirow{3}{*}{\rotatebox[origin=c]{90}{\textsc{imputation}}}& RMI & Tyler's estimator on multiple imputed data \cite{royston2004, vanbuuren2018} & yes \\
	& RSI & Tyler's estimator on stochastic imputed data & yes \\
	&Mean-Tyl & Tyler's estimator on data with mean imputation \cite{vanbuuren2018}& yes \\
\bottomrule
\end{tabular}
\caption{List of estimators used in numerical experiments. In the experiments, low-rank estimators keep their names with an additional ``-r'', \emph{e.g.}, EM-SCM-r.}
}
\label{tab:estimators}
\end{table*}

\subsection{Covariance estimation with missing data}
\label{subsec:num1}

This experiment shows the performances of covariance estimation as functions of the missing data ratio and missing data patterns. To compare the estimated covariance matrix to the true data covariance, the following geodesic distance is used \cite{bhatia2009positive}:

\begin{equation}
	\label{eq:distance}
	\delta_{\mathcal{S}^p_{++}}^2(\bm{\Sigma},\widehat{\bm{\Sigma}}) = ||\log(\bm{\Sigma}^{-\frac{1}{2}}\widehat{\bm{\Sigma}}\bm{\Sigma}^{-\frac{1}{2}})||_2^2
\end{equation}

This distance, which is sufficient to measure estimation errors, emanates from the Fisher metric for the Gaussian distribution on $\mathcal{S}_{++}^p$ \cite{bouchard2020}.

Three missing data patterns are examinated, which correspond to the following configurations (see Fig.~\ref{fig1}):
\begin{itemize}
	\item One block of missing data of size $(7\times 20)$. This is a special case of the monotone missing data pattern as studied in \cite{liu1999, liupalomar2019}. This configuration corresponds to the case where one group of variables is missing at the same observations, \emph{e.g.}, a group of sensors that are equally interrupted in time. This case is referred to as \textit{monotone pattern} in the experiments.
	\item Multiple blocks of missing data with various size and randomly scattered into the data set. This configuration corresponds to the general missing data pattern. It is most likely to happen in real-life applications (see, \emph{e.g.}, \cite{musial2011, shen2015}). This case is referred to as \textit{general pattern}.
	\item Randomly distributed missing data. This configuration is also a general missing data pattern except that multiple values for one variable at one observation are missing across the dataset. This case is referred to as \textit{random pattern}.
\end{itemize}

\begin{figure}
	\centering
	\includegraphics[scale=.9, trim={1.5cm 8cm 3cm 8cm}, clip]{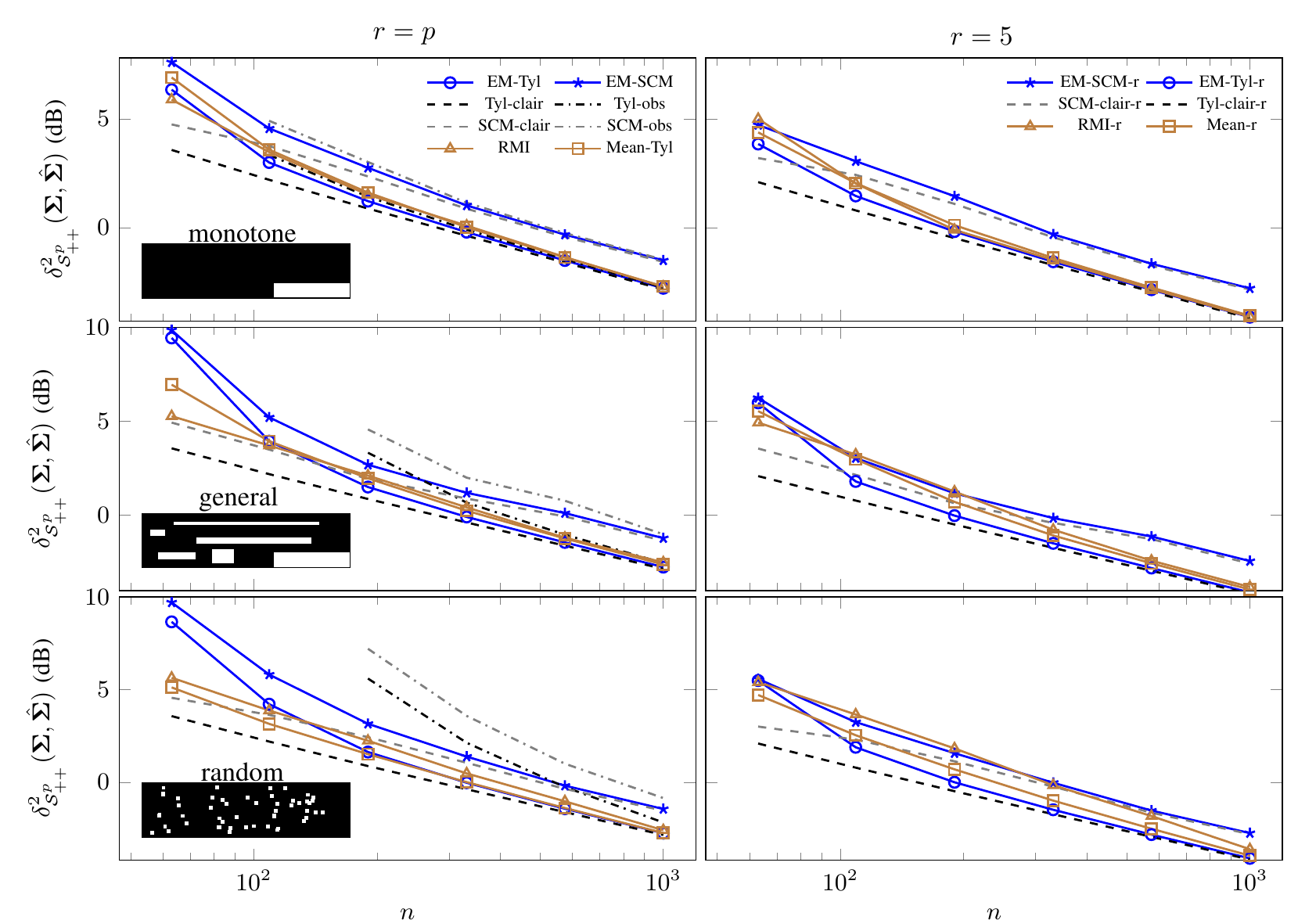}
	\caption{Mean of error measures~\ref{eq:distance} for $r=p$ (left) and $r=5$ (right) of methods EM-Tyl, EM-Tyl-r, EM-SCM, EM-SCM-r, Tyl-clair, Tyl-obs, SCM-clair, SCM-obs, RMI and mean imputation as functions of the number of samples $n$. The mean are computed over 500 simulated sets $\{\bm{y}_i\}_{i=1}^n$ for monotone (top), general (middle) and random (bottom) missing data pattern with $p=15$.}
	\label{fig3}
\end{figure}

\begin{figure}
	\centering
	\includegraphics[scale=.7]{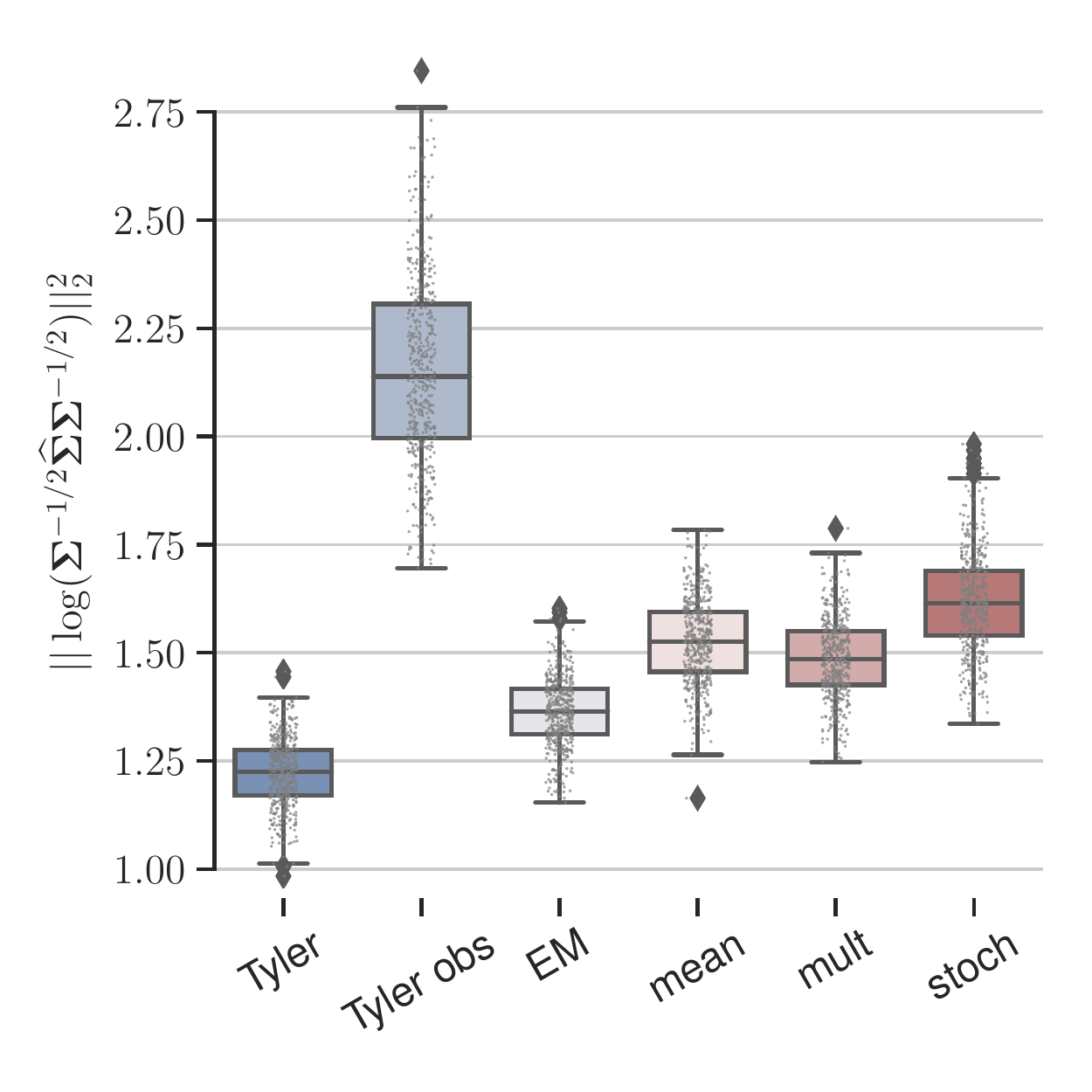}
	\caption{Boxplots showing geodesic distances of Tyler, Tyler observed, EM (ours), mean imputation, multiple imputation and stochastic imputation methods for $r=p$ and $n=200$. Small dots show the computed distances over 500 simulated sets $\{\bm{y}_i\}_{i=1}^n$ corrupted by 20\% of missing data with general pattern, which corresponds to the middle-left plot of Fig.~\ref{fig3}.}	
	\label{fig4}
\end{figure}

Fig.~\ref{fig3} and \ref{fig4} show the results in terms of mean distance for two rank settings (the full rank case $r=p$ and the low-rank case with $r=5$) and the three aforementioned pattern configurations. In the full rank case, EM-Tyl shows better performances than all other estimators, \blue{but fails when the number of available observations is small, especially when it gets close to $p$.} 

For the random pattern configuration, this result \blue{is even more notable and} when $n$ increases, EM-Tyl exhibits similar errors than mean imputation. Good performances of the latter can be explained by the increasing availability of observations as $n$ increases, which gives a better estimate of $\bm{\mu}^o_i$. Unsurprisingly, all estimators (except the ones based on the SCM) reach Tyler's estimator \blue{on clairvoyant data} as the missing data ratio decreases in large sample size. In the case of Gaussian estimators, the EM-SCM estimator shows very close performances than the clairvoyant SCM. As expected, Tyl-obs and SCM-obs display poor performances when the missing data ratio increases. 

Results for the low-rank model are illustrated for $r=5$. \blue{In all cases, we observe that EM-Tyl-r is a better estimator than EM-SCM-r in terms of distance. For small $n$, the gap between the EM-Tyl-r estimator and the imputation strategies (RMI-r and Mean-Tyl-r) diminishes compared to the full-rank case, and EM-Tyl-r significantly gives better estimates when the number of samples increases.}

\blue{In conclusion, depending on the available sample size and the missing data pattern at hand, different estimators are preferable: imputation strategies are more advantageous for small amount of observed samples whereas our estimator performs better as the number of observed samples grows and when the covariance matrix has a low-rank structure.}

\subsection{Should outliers be discarded?}
\label{subsec:num2}
Data corrupted by outliers is one of the main motivation of robust estimation \cite{zoubir2012}. In this case, one can be interested to know whether Tyler's estimator or EM-Tyl estimator gives higher performances. In the complete data case, the former can be directly estimated with outliers. To estimate the latter, outliers can be discarded or masked (\emph{i.e.}, set as missing data) if their position are known, which is a common approach in various applications. In the following, a data set $\{\bm{y}_i\}_{i=1}^n$ is drawn from the \blue{MSG} distribution with a fixed number of samples $n=200$ \blue{and with $p=15$}. Outliers are generated by adding white Gaussian noise (WGN) $\bm{z}_i \sim \mathcal{N}(\bm{0}, \sigma_{wgn}\blue{\bm{I}_p})$ with an increasing variance $\sigma_{wgn}$ on randomly selected observations $\bm{y}_i$.

Three data sets are considered to compare the CM estimation errors:
\begin{itemize}
	\item[1)] Corrupted data $\bm{y}_i + \bm{z}_i$;
	\item[2)] Masked data where the outliers are set as missing data.
	\item[3)] Data without outliers $\bm{y}_i$;
\end{itemize}

Tyler's estimator is computed on data set 1), whereas the EM-Tyl algorithm is applied to data set 2). Finally, the clairvoyant SCM and Tyler's estimators are computed from data set 3).

Fig.~\ref{fig5} gives an overlook on what choice of the CM estimation procedure would be preferable. It shows that it mainly depends on the outlier variance: for a small $\sigma_{wgn}$, Tyler's estimator should be preferred, whereas EM-Tyl is more suitable when the outlier variance reaches the variance of the signal $\sigma$. Note also that EM-Tyl provides better performances than the clairvoyant SCM until the missing data ratio reaches $\sim$45\% of the data.
\begin{figure}
	\centering
	\includegraphics[scale=.62, trim={.1cm .1cm 1.4cm 1.4cm}, clip]{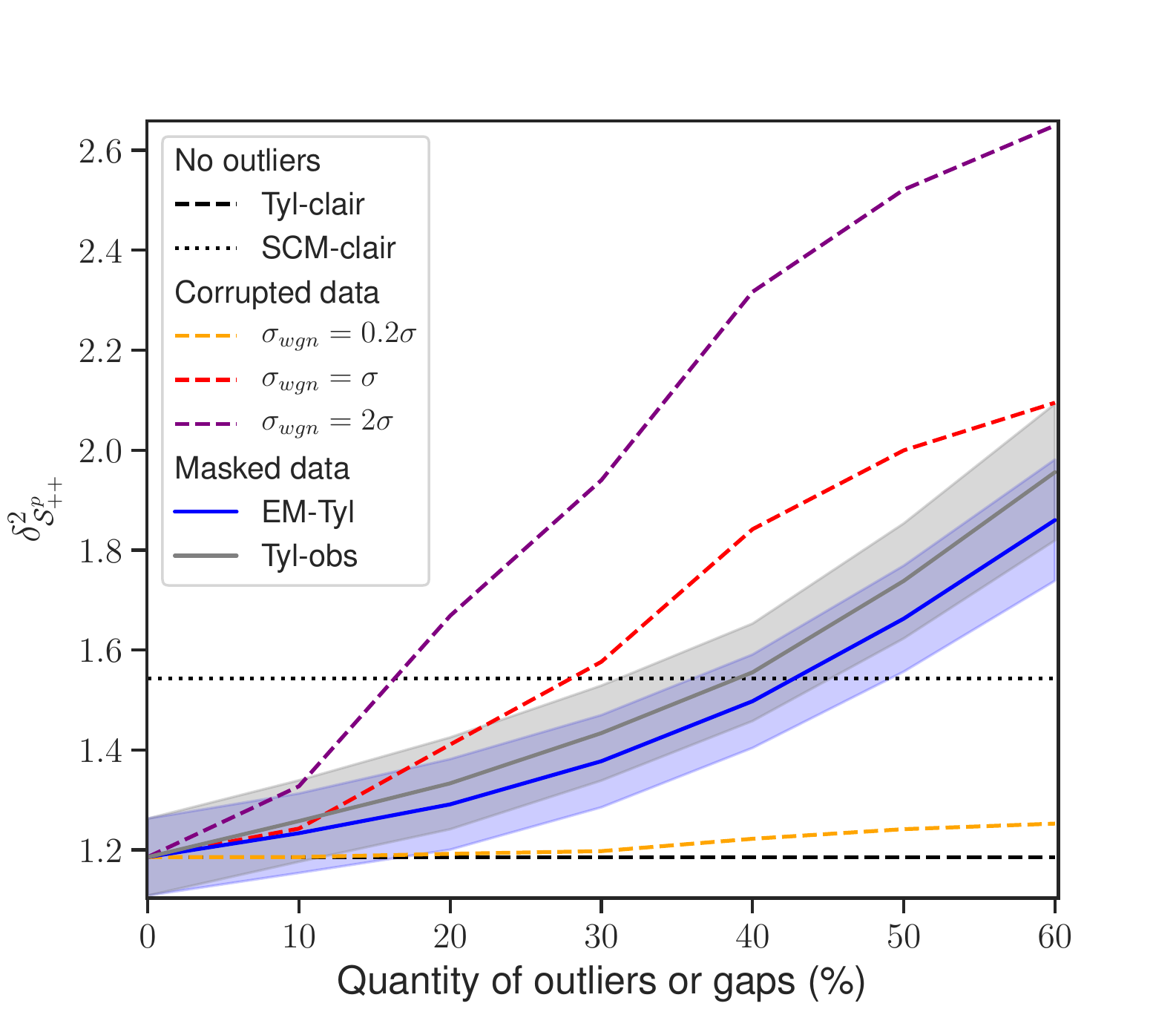}
	\caption{Mean geodesic distances over 200 simulated sets $\{\bm{y}_i\}_{i=1}^n$ as function of the missing data or outlier ratio (\%) for $n=200$. Tyl-clair and SCM-clair errors are shown with no outliers. Tyl-clair is displayed on corrupted data with different outlier variance $\sigma_{wgn}$ w.r.t. the data variance $\sigma$. EM-Tyl and Tyl-obs errors are shown when outliers are masked as missing values.}
	\label{fig5}
\end{figure}

\subsection{Robust imputation of missing values with outliers}
\label{subsec:num3}

In this experiment, we propose to apply our estimators to missing data imputation with possible outliers corrupting the data. Missing data imputation \cite{vanbuuren2018} concerns a wide range of applications, including remote sensing \cite{shen2015}. 
A recent procedure to deal with missing values was developed, namely the EM-EOF method \cite{hf2020}, which uses the EM algorithm and empirical orthogonal functions (EOFs) to iteratively decompose the CM and reconstruct the incomplete data with a few number of selected EOFs $k\ll p$. The final number of components is the one giving the minimal error between the initial data and the imputed data. As shown in Fig.~\ref{fig6}, instead of the SCM (which is used in the aforementioned study), the EM-Tyl-r estimator is used with $r=k$ at the last iteration of the EM-EOF algorithm. This is justified by the fact that this algorithm iteratively updates the SCM with the predicted missing values at each step, whereas the EM-Tyl-r estimator only needs the missing data pattern, which remains the same at each iteration. The RMI estimator~(\ref{eq:rmi}) is also considered for comparison.

The data is generated using a Haystack-type model \cite{lerman2012robust} which draws samples $\{\bm{y}_i\}_{i=1}^n$ as inliers $\bm{y}_i^{\text{in}}$ and outliers $\bm{y}_i^{\text{out}}$:
\begin{align}
\label{eq:haystack}
\{\bm{y}_i\}_{i=1}^n &= \{ \{\bm{y}_i^{\text{in}} \}_{i\in\{1, \dots, n_{\text{in}}\}}, \{ \bm{y}_i^{\text{out}}\}_{i\in\{1, \dots, n_{\text{out}}\}} \} \\
\bm{y}_i^{\text{in}} &\sim \mathcal{N}(\bm{0},\bm{I}_p + \sigma_s^2\bm{U}\bm{U}^{\top}) \\
\bm{y}_i^{\text{out}} &\sim \mathcal{N}(\bm{0},\bm{I}_p + \sigma_o^2\bm{U}_{\bot}\bm{U}_\bot^{\top})
\end{align}
where $\bm{U} \subset \text{St}_{p,k}$ is the signal subspace basis, $\sigma_s^2$ and $\sigma_o^2$ are respectively the signal to noise ratio (SNR) and the outlier to noise ratio (ONR), and $100\times\frac{(n-n_{\text{in}})}{n}$ is the outlier ratio in the data set in percentage. 

In the experiments, we choose $\sigma_s^2 = 10$, $\sigma_o^2\in\{0,1.5,3,\dots,30\}$ and the outlier ratio varies from 0 to 50\% of the data. Note that inliers and outliers are whole vectors (not just entries) which are randomly chosen in the data set. 30\% of the data is discarded as missing values under a general pattern.

\begin{figure}
	\centering
	\includegraphics[scale=0.8, trim={3.7cm .1cm 2.7cm .1cm}, clip]{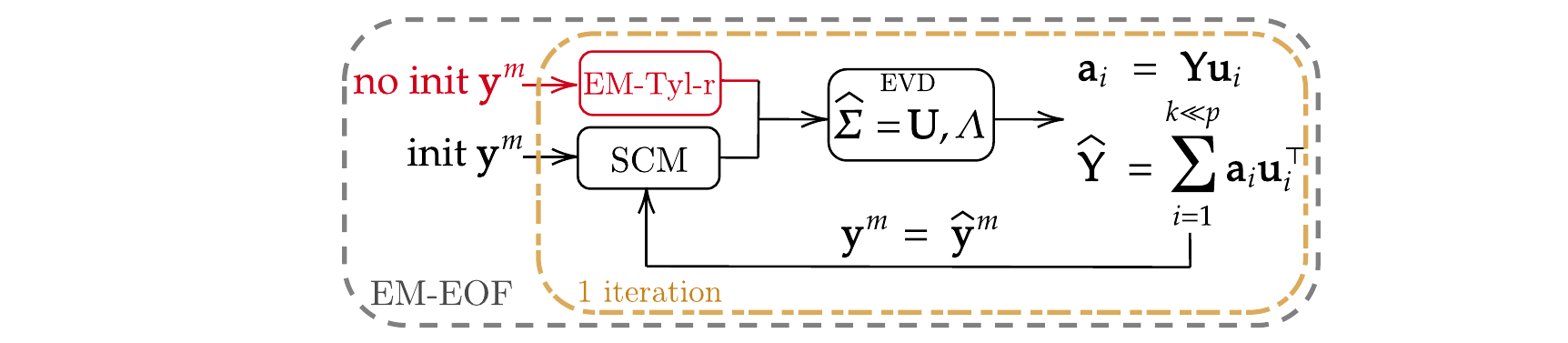}
	\caption{Diagram of the EM-EOF method \cite{hf2020} for missing data imputation. After initializing the missing values, the CM matrix is estimated from the data matrix $\bm{Y}$, which is then reconstructed with a number of components $k\ll p$. Instead of using the SCM, the EM-Tyl-r estimator with $r=k$ is plugged in at the last iteration of the algorithm (in red).}
	\label{fig6}
\end{figure}

To measure the performance of imputation, \blue{a set $\mathcal{Y} = \{y_l\}_{l=1}^s$ consisting of $s$ random points over $\{y_{ij}\}_{1\le i \le n, 1\le j\le p}$ is selected. $\mathcal{Y}$ is called the cross-validation (CV) data set.} Then, we compute the root-mean-square error (RMSE) between $\mathcal{Y}$ and its reconstruction $\widehat{\mathcal{Y}} = \{\hat{y}_l\}_{l=1}^s$ after imputation:

\begin{equation} 
	\label{eq:rmse}
	\delta_{imputation} = \bigg[\frac{1}{s}\sum_{l=1}^{s}  (\hat{y}_l  - y_l)^2\bigg]^{1/2}
\end{equation}
\blue
In total, 1\% of the data is chosen for the cross-validation data set. These values are artificially removed and copied before the process, and compared to the new values after the imputation procedure. For model~(\ref{eq:haystack}), CV errors illustrated by Fig.~\ref{fig7} show a substantial gain by replacing the SCM by the EM-Tyl-r estimator at the last step of the EM-EOF method, whereas RMI performs poorly. The gain is larger for important signal to outlier ratio $\sigma_{o}^2$, which confirms the robustness of the proposed estimator for robust low-rank imputation under model~(\ref{eq:haystack}).

\begin{figure}
	\centering
	\begin{tikzpicture}
	\node[inner sep=0pt] at (0,0)
	{\includegraphics[trim={1.4cm .4cm .5cm 8.7cm}, clip, scale=.7]{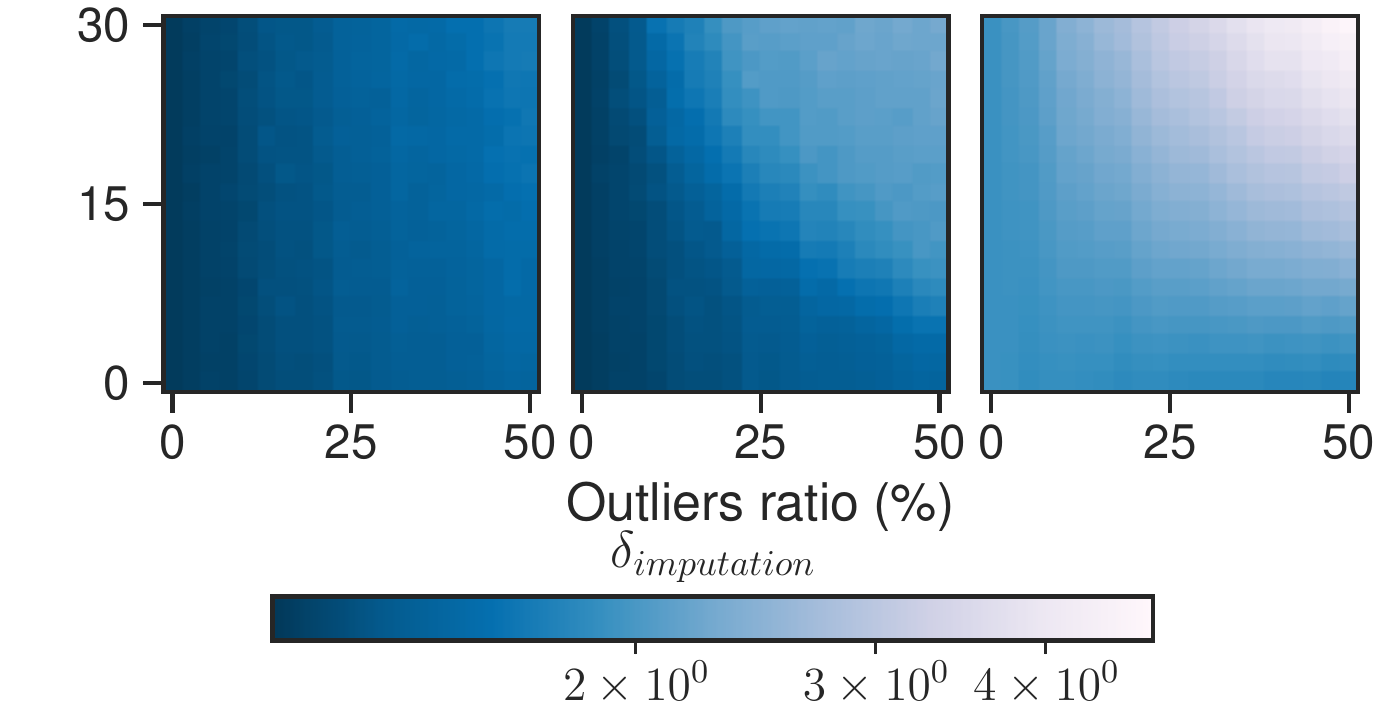}};
	\node[rotate=90] at (-4.8,1)  {$\sigma_{o}^2$};
	\draw[gray, dashed] (-3.2,2.45) rectangle (1,3.2);
	\node[gray] at (-1,3)  {\small{EM-EOF}};
	\node at (-2.4,2.65)  {EM-Tyl-r};	
	\node at (0.4,2.65)  {SCM};
	\node at (3.4,2.65)  {Robust MI};
	\end{tikzpicture}
	\caption{Mean of cross-validation error~(\ref{eq:rmse}) over 200 simulated sets $\{\bm{y}_i\}_{i=1}^n$ as function of the outlier ratio (\%) and signal to outlier ratio $\sigma_{o}^2$ for the EM-EOF using EM-Tyl-r estimator, EM-EOF using the SCM and RMI procedures. 30 \% of missing data under the general pattern are generated. Chosen parameters are $n=200$ and $p=15$.}
	\label{fig7}
\end{figure}
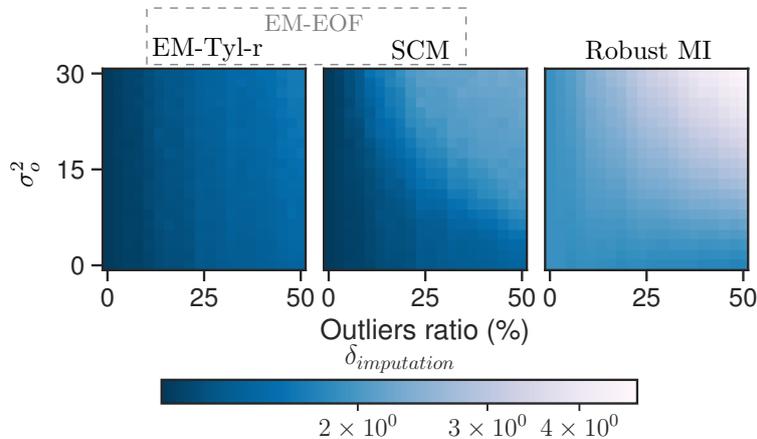

\section{Application to classification and clustering problems}
\label{sec5}

\subsection{Classification with missing values}

In remote sensing, and more particularly in multispectral and hyperspectral imaging, missing data can arise for various reasons including:
\begin{itemize}
	\item[1)] Clouds when the sensor operates in the visible part of the spectrum \cite{shen2015};
	\item[2)] Sensor failure, as the known problem on the scan-line corrector images of the Landsat ETM+ sensor \cite{chen2011} or Aqua MODIS band 6 \cite{rakwatin2008}; 
\end{itemize}

In these cases, the data (or the entire band itself) is masked, downsampled to avoid temporal gaps or restored using specific gap-filling methods \cite{chen2011,rakwatin2008}. When it comes to classification tasks, incomplete data is a challenge. In this scope, existing approaches classify, mask, and interpolate values from cloudy observations in a pre-processing step \cite{breizhcrops2020}.

To tackle this challenge, the proposed estimators are used as a set of descriptors\footnote{In machine learning problems, statistical descriptors are classically used as they are often more discriminative than raw data.} $\{\theta_i\}$ to a classification problem on the \textit{BreizhCrops} data set \cite{breizhcrops2020}. The data consists of Sentinel-2 multivariate time series of field crop reflectances on the Brittany region over 23 spectral bands. It is divided into four sub-regions called FRH01, FRH02, FRH03 and FRH04 corresponding to the four departments of Brittany. These sub-regions contain field labels which are gathered in 9 selected classes representing crop categories, \emph{e.g.}, \textit{barley}, \textit{nuts}, \textit{wheat}, etc. Following \cite{breizhcrops2020}, each band is mean-aggregated over one field parcel to a feature vector $\bm{y}^t \in \mathbb{R}^p$, where $p$ is the number of features (here the spectral bands) and $t$ a timestamp. Thus, the whole data set is denoted $\{\{\bm{y}_k^t\}_{k=1}^K\}_{t=1}^T$, which corresponds to the aggregation of all reflectances at field parcels $k\in [1,K]$ and timestamps $t\in[1,T]$. The experiment is done on the L1C top-of-atmosphere product and on 13 bands selected by the \textit{BreizhCrops} processing chain \cite{breizhcrops2020}. \blue{For the L1C product, the total number of parcels is $K=608263$ and the number of observations per parcel is $T=51$.}

A supervised covariance-based classification is performed using a Minimum distance to Riemannian mean (MDRM) classifier \cite{barachant2012}, which works in a test-training form. The training step provides a set of $p\times p$ SPD matrices encoding field parcels for the available classes. Assuming that the data is complete, the set of SCMs would be computed as:

\begin{equation}
	\bm{\Sigma}_k = \frac{1}{T}\sum_{t=1}^{T}\bm{y}^t_k(\bm{y}^t_k)^\top, \quad k=1,\dots,K
\end{equation}
For each class, a center of mass of the available parcels is estimated. In the test step, a field parcel is also encoded as an SPD matrix and then assigned to the class whose center of mass is the closest according to the distance (\ref{eq:distance}) acting on the manifold $\mathcal{S}^p_{++}$ \cite{congedo2019}.

Before the experiment, missing data are artificially created in the data set. For this, successive bands are set as missing and the performances are evaluated as function of the missing data ratio. The classifier is trained first on the \blue{FRH01 region and tested on the FRH03 region. A second experiment is conducted, where the train data set is changed to the FRH04 region while the test set remains unchanged.}

Results in terms of overall accuracy (OA) for the FRH03 region \blue{with two different training sets} versus the number of missing bands are displayed in Fig.~\ref{fig:8}. Three covariance estimation strategies are shown: covariance estimation for Gaussian data (EM-SCM), robust stochastic imputation (RSI) and robust covariance estimator (EM-Tyl) multiplied by the \blue{geometric} mean of the estimated scales $\widehat{\tau}_i$:
\begin{equation}
	\text{EM-Tyl-}\Pi\tau_i = \text{EM-Tyl} \times \bigg(\prod_{i=1}^{T}\widehat{\tau}_i\bigg)^{\frac{1}{T}}
\end{equation}
\blue{Here, the geometric mean acts as an additional power information to the EM-Tyl estimator and better fits the average of textures that might vary on different scales.}
Unfortunately, experiments have shown that the low-rank structure does not improve the OA, which might be due to the mean-aggregation over each parcel, which essentially acts as a filter. 

Results show that, for this data set, the MDRM classifier based on robust covariance matrix estimation (EM-Tyl-$\Pi\tau_i$-based) is \blue{generally} more robust to missing bands than the EM-SCM one, whereas the RSI-based classifier gives a rapidly decreasing OA when the gaps ratio increases. \blue{Furthermore, we observe that the EM-SCM-based classifier is more suited for large missing data ratio when the training phase is performed on the FRH04 region}. Interestingly, the classifier based on the EM-Tyl-$\Pi\tau_i$ estimator \blue{provides} an almost stable OA until the gap ratio reaches \blue{$28\%$} of the data \blue{(4 missing bands) when when FRH01 is the training set and $14\%$ of the data (2 missing bands) when the classifier is trained on FRH04}.

\begin{figure}
	\centering
	\includegraphics[trim={4.5cm 11cm 3.5cm 11.1cm}, clip, scale=.78]{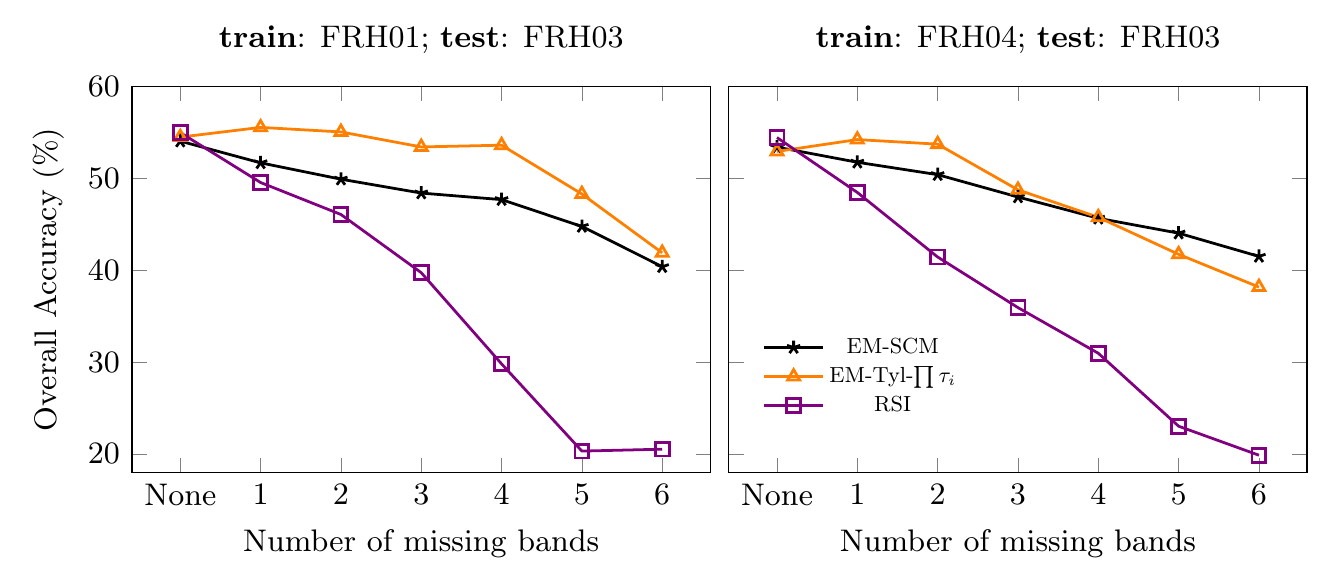}
	\caption{Classification mean overall accuracy (\%) versus the number of randomly missing bands  (1 band $\sim 7\%$ of the data) over 50 run of the MDRM \cite{barachant2012} classifier based on EM-SCM, EM-Tyl-$\prod\tau_i$ and RSI estimators.}
	\label{fig:8}
\end{figure}

\subsection{Image clustering with missing values}

We apply the proposed CM estimators to a hyperspectral image clustering problem using the \textit{K-means++} algorithm \cite{vassilvitskii2006} on the \textit{Indian Pines} data set \cite{baumgardner2015}, consisting of a $145\times 145$ pixels image with $p=200$ bands. As in the classification problem, the proposed estimators EM-SCM and EM-Tyl are used as descriptors $\{\theta_i\}$, as well as EM-SCM-r and EM-Tyl-r. The aim is to partition the descriptors in a number of clusters which correspond to the 16 classes dividing the \textit{Indian Pines} image.

The image is first centered by subtracting the global mean. Then, a sliding window of size $w \times w$ is applied to the image. One descriptor $\{\theta_i\}$ is estimated using the $n = w^2$ observations in each window denoted $\bm{X}_i \in \mathbb{R}^{p\times n}$. Thus, we obtain a set of descriptors $\{\theta_i\}$ to cluster using a \textit{K-means++} \cite{vassilvitskii2006}. For the descriptors using the low-rank model~(\ref{eq:lowrank}), the first $r=5$ components are kept which concur with the five principal eigenvectors of the SCM of \textit{Indian Pines} containing more than 95\% of the total cumulative variance\footnote{This measure can be easily computed by the formula $\sum_{i=1}^{r}\lambda_i/\sum_{i=1}^{p}\lambda_i$.}.

Due to the high dimensionality and a possible runtime overflow due to the creation of missing data, the image is subsampled regularly to reduce the dimension to 20 bands and cropped to get a final $85 \times 70$ image composed by 5 or the 16 original classes (see Fig.~\ref{fig:9a}). As shown by Fig.~\ref{fig:9b}, sensor failure is simulated by adding missing values on random columns of selected bands.

\begin{figure}
	\centering
	\begin{subfigure}{0.49\linewidth}
		\centering
		\includegraphics[scale=0.5, trim={1cm .5cm 1cm .5cm}, clip]{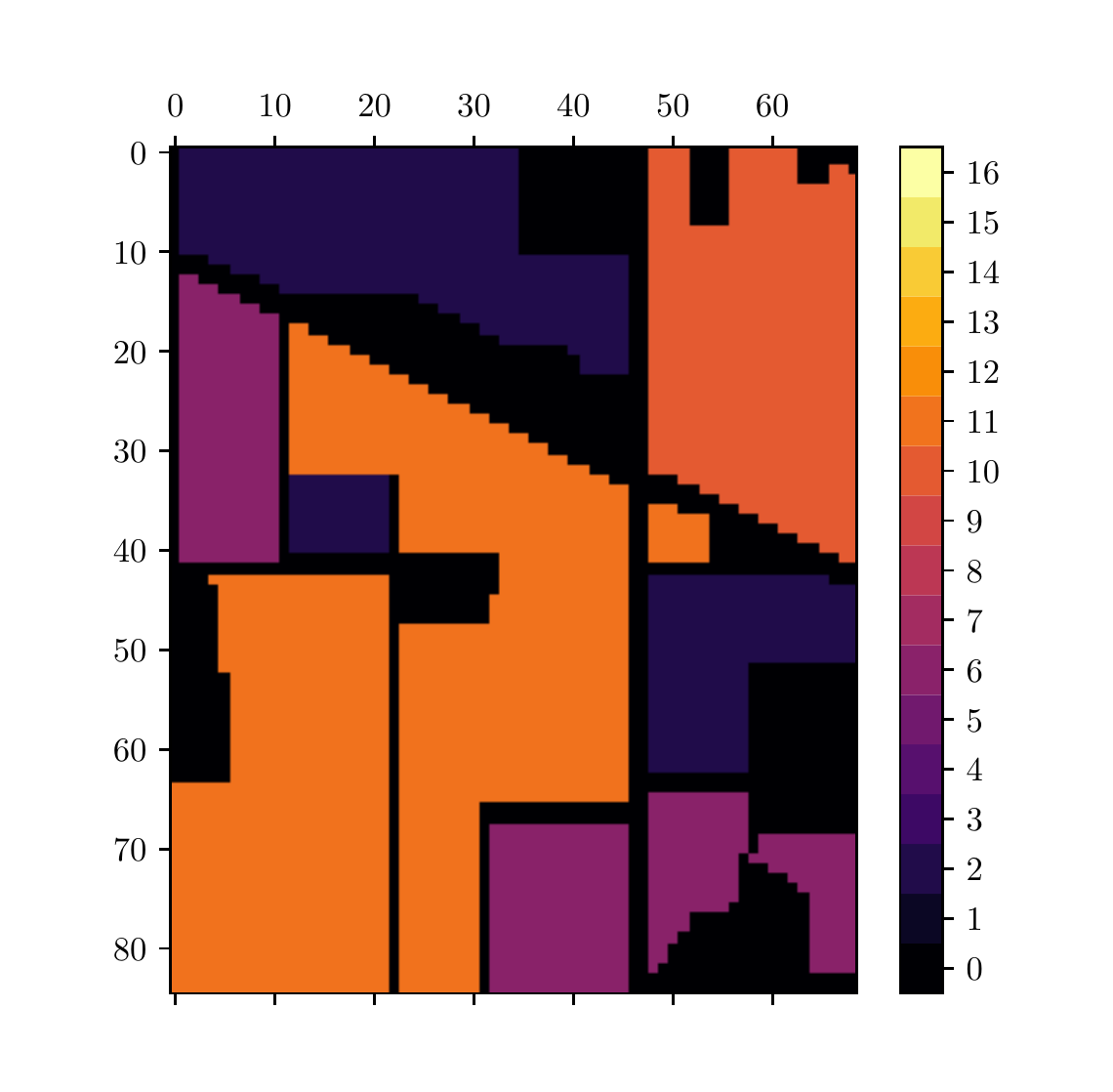}	
		\caption{Ground truth}
		\label{fig:9a}
	\end{subfigure}
	\begin{subfigure}{0.49\linewidth}
		\centering
		\includegraphics[scale=0.5, trim={1cm .5cm .5cm .5cm}, clip]{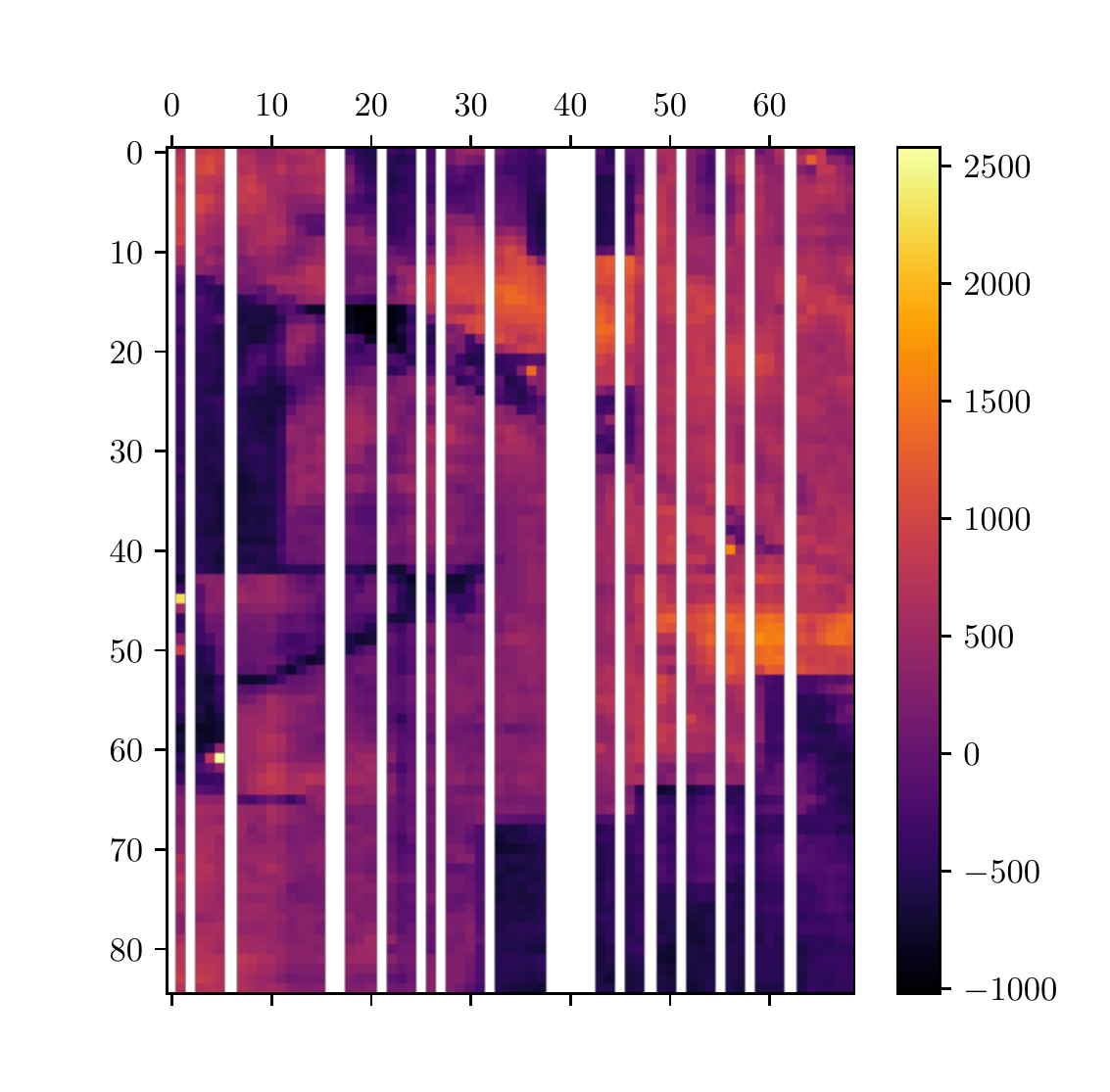}
		\caption{Simulated sensor failure}	
		\label{fig:9b}
	\end{subfigure}
	\caption{Selected sub-image of \textit{Indian Pines}: ground truth (left) and simulated sensor failure on band 10 (right) where white stripes are the missing values (20 columns of pixels are set as missing).}
	\label{fig:9}
\end{figure}

The averaged Overall Accuracy (OA) are reported in Fig.~\ref{fig:10}. We observe that in general, the descriptors based on RSI and RSI-r estimators give lower accuracies compared to the descriptors using the SCM, SCM-r, EM-Tyl and EM-Tyl-r estimators. As the missing data ratio increases, the performances of the EM-SCM estimator are undermined whereas EM-Tyl estimator accuracies remain stable \blue{around 50\%}. For large missing data ratio, EM-SCM-r performs surprisingly well compared to EM-Tyl and EM-Tyl-r. For lower missing data ratios, clustering achieves its best performances using the EM-Tyl-r estimator.

\begin{figure}
	\centering
	\begin{tikzpicture}
	\begin{axis}[
	at={(0,0)},
	width=.8\linewidth,
	height=5.5cm,
	xlabel={\# of incomplete bands},
	xticklabels={None,None,1,2,3,4,5},
	ylabel={Overall Accuracy},
	legend columns=2,
	legend entries={EM-SCM, EM-SCM-r,
		EM-Tyl, EM-Tyl-r,
		RSI, RSI-r
	},
	legend style={nodes={scale=.8},at={(.75,.92)},anchor=north,draw=none}
	]
	\addplot[thick, mark=star, mark size=2pt] table [x=Band, y=SCM] {./OA_indian_pines_low_rank};
	\addplot[thick, mark=star, dashed, mark options={solid}, mark size=2pt] table [x=Band, y=SCM_lr] {./OA_indian_pines_low_rank};
	\addplot[thick, mark=triangle, gray, mark size=2pt] table [x=Band, y=Tyler] {./OA_indian_pines_low_rank};
	\addplot[thick, mark=triangle, dashed, mark options={solid}, gray, mark size=2pt] table [x=Band, y=Tyler_lr] {./OA_indian_pines_low_rank};
	\addplot[thick, mark=square, violet, mark size=2pt] table [x=Band, y=Imput] {./OA_indian_pines_low_rank};
	\addplot[thick, mark=square, dashed, mark options={solid}, violet, mark size=2pt] table [x=Band, y=Imput_lr] {./OA_indian_pines_low_rank};
	\end{axis}
	\end{tikzpicture}
	\caption{Clustering mean overall accuracy (\%) versus the number of incomplete bands (as in Fig.~\ref{fig:9b}) over 10 run of the \textit{Kmeans++} algorithm for the descriptors based on EM-SCM, EM-Tyl, RSI, EM-SCM-r, EM-Tyl-r and RSI-r estimators.}
	\label{fig:10}
\end{figure}
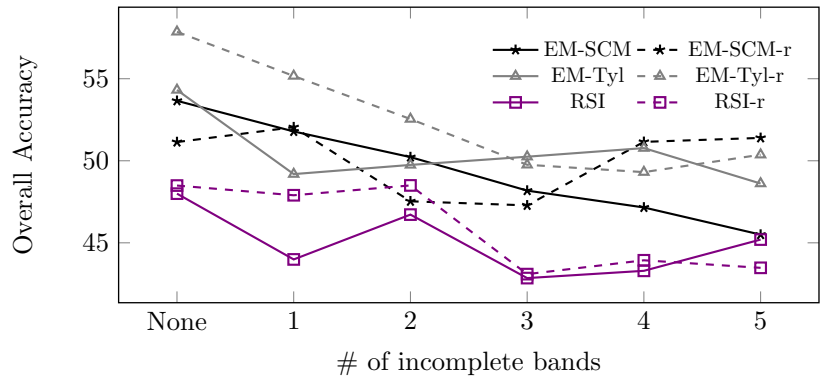

\section{Conclusion}
\label{sec6}
This article proposes a novel procedure based on the EM algorithm to perform robust low-rank estimation of the covariance matrix with missing data following a general pattern. The developed tools take advantage of the \blue{MSG} distribution and of the LR structure of the covariance matrix. Closed-form expressions of the unknown CM and associated scales are derived and integrated to the M-step of the algorithm, which generalizes the Gaussian case. The performance of the proposed estimators are validated on simulated data sets with missing values and corrupted by outliers, as well as real-world incomplete data sets. Compared to the classical Gaussian assumption and to the unstructured model, experiments show the possibility to improve the results in terms of CM estimation and robust imputation, as well as supervised (classification) and unsupervised (image clustering) problems. Some potential extensions of this work include the extension to other classes of \blue{compound Gaussian} distributions and experimenting classification tasks with temporal gaps rather than spectral gaps, as well as more spectral information by reducing the downsampling rate. 


\section*{Acknowledgment}

\small This work was supported by ANR ASTRID project MARGARITA (ANR-17-ASTR-0015). The authors would like to thank Antoine Collas and Ammar Mian for their precious help on the implementation pipeline of the MDRM and $Kmeans++$ algorithms for real data, and Florent Bouchard for his deep insights in robust estimation and covariance-based classification.



\bibliographystyle{elsarticle-num}
\bibliography{biblio}

\begin{thebibliography}{10}
\expandafter\ifx\csname url\endcsname\relax
  \def\url#1{\texttt{#1}}\fi
\expandafter\ifx\csname urlprefix\endcsname\relax\def\urlprefix{URL }\fi
\expandafter\ifx\csname href\endcsname\relax
  \def\href#1#2{#2} \def\path#1{#1}\fi

\bibitem{anderson1957}
T.~W. Anderson, Maximum likelihood estimates for a multivariate normal
  distribution when some observations are missing, Journal of the american
  Statistical Association 52~(278) (1957) 200--203.

\bibitem{afifi1966}
A.~A. Afifi, R.~M. Elashoff, Missing observations in multivariate statistics
  {I}. {R}eview of the literature, Journal of the American Statistical
  Association 61~(315) (1966) 595--604.

\bibitem{little1987}
R.~J.~A. Little, D.~B. Rubin, Statistical analysis with Missing Data, Wiley,
  New York, 1987.

\bibitem{schafer1997}
J.~L. Schafer, Analysis of incomplete multivariate data, CRC press, 1997.

\bibitem{benzecri1973}
J.-P. Benz{\'e}cri, et~al., L'analyse des donn{\'e}es, Vol.~2, Dunod Paris,
  1973.

\bibitem{vanbuuren2018}
S.~Van~Buuren, Flexible imputation of missing data, CRC press, 2018.

\bibitem{walczak1}
B.~Walczak, D.~L. Massart, Dealing with missing data : {P}art {I}, Chemom.
  Intell. Lab. Syst. 58 (2001) 15--27.

\bibitem{shen2015}
H.~Shen, X.~Li, Q.~Chen, C.~Zeng, G.~Yang, H.~Li, L.~Zhang, Missing information
  reconstruction of remote sensing data: {A} technical review, IEEE Geosci.
  Remote Sens. Mag. 3 (2015) 61--85.
\newblock \href {https://doi.org/10.1109/MGRS.2015.2441912}
  {\path{doi:10.1109/MGRS.2015.2441912}}.

\bibitem{rubin1976}
D.~B. Rubin, Inference and missing data, Biometrika 63~(3) (1976) 581--592.

\bibitem{Dempster1977}
A.~P. Dempster, N.~M. Laird, D.~B. Rubin, Maximum likelihood from incomplete
  data via the {EM} algorithm, J. Royal Statistical Society. Series B
  (Methodological) 39~(1) (1977) 1--38.

\bibitem{jamshidian1997}
M.~Jamshidian, An {EM} algorithm for {ML} factor analysis with missing data,
  in: Latent variable modeling and applications to causality, Springer, 1997,
  pp. 247--258.

\bibitem{liu1999}
C.~Liu, Efficient {ML} estimation of the multivariate normal distribution from
  incomplete data, Journal of Multivariate Analysis 69 (1999) 206--217.
\newblock \href {https://doi.org/10.1006/jmva.1998.1793}
  {\path{doi:10.1006/jmva.1998.1793}}.

\bibitem{schneider2001}
T.~Schneider, Analysis of incomplete climate data: {E}stimation of mean values
  and covariance matrices and imputation of missing values, J. Climate 14
  (2001) 853--871.

\bibitem{lounici2014}
K.~Lounici, et~al., High-dimensional covariance matrix estimation with missing
  observations, Bernoulli 20~(3) (2014) 1029--1058.

\bibitem{stadler2014}
N.~St{\"a}dler, D.~J. Stekhoven, P.~B{\"u}hlmann, Pattern alternating
  maximization algorithm for missing data in high-dimensional problems., J.
  Mach. Learn. Res. 15~(1) (2014) 1903--1928.

\bibitem{chen2009}
T.~Chen, E.~Martin, G.~Montague, Robust probabilistic {PCA} with missing data
  and contribution analysis for outlier detection, Computational Statistics \&
  Data Analysis 53~(10) (2009) 3706--3716.

\bibitem{josse2012}
J.~Josse, F.~Husson, Handling missing values in exploratory multivariate data
  analysis methods, Journal de la Soci{\'e}t{\'e} Fran{\c{c}}aise de
  Statistique 153~(2) (2012) 79--99.

\bibitem{little1988}
R.~J. Little, Robust estimation of the mean and covariance matrix from data
  with missing values, Journal of the Royal Statistical Society: Series C
  (Applied Statistics) 37~(1) (1988) 23--38.

\bibitem{liu1995}
C.~Liu, D.~B. Rubin, {ML} estimation of the t distribution using {EM} and its
  extensions, {ECM} and {ECME}, Statistica Sinica (1995) 19--39.

\bibitem{liupalomar2019}
J.~Liu, D.~P. Palomar, Regularized robust estimation of mean and covariance
  matrix for incomplete data, Signal Processing 165 (2019) 278 -- 291.
\newblock \href {https://doi.org/https://doi.org/10.1016/j.sigpro.2019.07.009}
  {\path{doi:https://doi.org/10.1016/j.sigpro.2019.07.009}}.

\bibitem{tyler1987}
D.~E. Tyler, A distribution-free {M}-estimator of multivariate scatter, The
  Annals of Statistics 15~(1) (1987) 234--251.

\bibitem{conte2002}
E.~Conte, A.~De~Maio, G.~Ricci, Recursive estimation of the covariance matrix
  of a compound-{G}aussian process and its application to adaptive {CFAR}
  detection, IEEE Transactions on signal processing 50~(8) (2002) 1908--1915.

\bibitem{pascal2008}
F.~Pascal, Y.~Chitour, J.-P. Ovarlez, P.~Forster, P.~Larzabal, Covariance
  structure maximum-likelihood estimates in compound {G}aussian noise:
  {E}xistence and algorithm analysis, IEEE Transactions on Signal Processing
  56~(1) (2008) 34--48.
\newblock \href {https://doi.org/10.1109/TSP.2007.901652}
  {\path{doi:10.1109/TSP.2007.901652}}.

\bibitem{zoubir2018}
A.~M. Zoubir, V.~Koivunen, E.~Ollila, M.~Muma, Robust statistics for signal
  processing, Cambridge University Press, 2018.

\bibitem{gini2000}
F.~Gini, M.~V. Greco, M.~Diani, L.~Verrazzani, Performance analysis of two
  adaptive radar detectors against non-gaussian real sea clutter data, IEEE
  Transactions on Aerospace and Electronic Systems 36~(4) (2000) 1429--1439.

\bibitem{greco2007}
M.~S. Greco, F.~Gini, Statistical analysis of high-resolution {SAR} ground
  clutter data, IEEE Transactions on Geoscience and Remote sensing 45~(3)
  (2007) 566--575.

\bibitem{theiler2005}
J.~Theiler, B.~Foy, A.~Fraser, Characterizing non-gaussian clutter and
  detecting weak gaseous plumes in hyperspectral imagery, Proceedings SPIE 5806
  (2005) 182--193.

\bibitem{formont2010}
P.~Formont, F.~Pascal, G.~Vasile, J.-P. Ovarlez, L.~Ferro-Famil, Statistical
  classification for heterogeneous polarimetric {SAR} images, IEEE Journal of
  selected topics in Signal Processing 5~(3) (2010) 567--576.

\bibitem{ollier2018}
V.~Ollier, M.~N. El~Korso, A.~Ferrari, R.~Boyer, P.~Larzabal, Robust
  distributed calibration of radio interferometers with direction dependent
  distortions, Signal Processing 153 (2018) 348--354.

\bibitem{abdallah2020}
R.~B. Abdallah, A.~Breloy, M.~N. El~Korso, D.~Lautru, {Bayesian signal subspace
  estimation with compound Gaussian sources}, Signal Processing 167 (2020)
  107310.

\bibitem{ollila2012}
E.~Ollila, D.~E. Tyler, V.~Koivunen, H.~V. Poor, Complex elliptically symmetric
  distributions: {S}urvey, new results and applications, IEEE Transactions on
  signal processing 60~(11) (2012) 5597--5625.

\bibitem{frahm2010}
G.~Frahm, U.~Jaekel, A generalization of {T}yler’s {M}-estimators to the case
  of incomplete data, Computational Statistics \& Data Analysis 54~(2) (2010)
  374--393.

\bibitem{little2019}
R.~J. Little, D.~B. Rubin, Statistical analysis with missing data, Vol. 793,
  John Wiley \& Sons, 2019.

\bibitem{larsson2001}
E.~G. Larsson, P.~Stoica, High-resolution direction finding: the missing data
  case, IEEE transactions on signal processing 49~(5) (2001) 950--958.

\bibitem{little2002}
R.~J.~A. Little, D.~B. Rubin, Statistical analysis with Missing Data, 2nd
  Edition, Wiley, New York, 2002.

\bibitem{goodman2007}
N.~Goodman, J.~Stiles, On clutter rank observed by arbitrary arrays., IEEE
  Transactions on Signal Processing 55 (2007) 178--186.

\bibitem{wiesel2012}
A.~Wiesel, Unified framework to regularized covariance estimation in scaled
  {G}aussian models, IEEE Transactions on Signal Processing 60~(1) (2012)
  29--38.
\newblock \href {https://doi.org/10.1109/TSP.2011.2170685}
  {\path{doi:10.1109/TSP.2011.2170685}}.

\bibitem{yao1973}
K.~Yao, A representation theorem and its applications to spherically-invariant
  random processes, IEEE Transactions on Information Theory 19~(5) (1973)
  600--608.

\bibitem{maronna2019}
R.~A. Maronna, R.~D. Martin, V.~J. Yohai, M.~Salibi{\'a}n-Barrera, Robust
  statistics: theory and methods (with {R}), John Wiley \& Sons, 2019.

\bibitem{mian2018}
A.~Mian, G.~Ginolhac, J.-P. Ovarlez, A.~M. Atto, New robust statistics for
  change detection in time series of multivariate {SAR} images, IEEE
  Transactions on Signal Processing 67~(2) (2018) 520--534.

\bibitem{ruppert2011}
D.~Ruppert, D.~S. Matteson, Statistics and data analysis for financial
  engineering, Vol.~13, Springer, 2011.

\bibitem{johnstone2001}
I.~M. Johnstone, {On the distribution of the largest eigenvalue in principal
  components analysis}, The Annals of Statistics 29~(2) (2001) 295 -- 327.

\bibitem{stoica2004}
P.~Stoica, Y.~Selen, Model-order selection: a review of information criterion
  rules, IEEE Signal Processing Magazine 21~(4) (2004) 36--47.
\newblock \href {https://doi.org/10.1109/MSP.2004.1311138}
  {\path{doi:10.1109/MSP.2004.1311138}}.

\bibitem{tipping1999}
M.~E. Tipping, C.~M. Bishop, Probabilistic principal component analysis,
  Journal of the Royal Statistical Society. Series B (Statistical Methodology)
  61~(3) (1999) 611--622.

\bibitem{aubry2021}
A.~Aubry, A.~De~Maio, S.~Marano, M.~Rosamilia, Structured covariance matrix
  estimation with missing-(complex) data for radar applications via
  expectationmaximization, IEEE Transactions on Signal Processing (2021)
  1--1\href {https://doi.org/10.1109/TSP.2021.3111587}
  {\path{doi:10.1109/TSP.2021.3111587}}.

\bibitem{paindaveine2008}
D.~Paindaveine, A canonical definition of shape, Statistics \& probability
  letters 78~(14) (2008) 2240--2247.

\bibitem{zhang2017}
X.~Zhang, M.~N. El~Korso, M.~Pesavento, {MIMO radar target localization and
  performance evaluation under SIRP clutter}, Signal Processing 130 (2017)
  217--232.

\bibitem{anderson1965}
T.~Anderson, An Introduction to Multivariate Statistical Analysis, Wiley, New
  York, 1965.

\bibitem{kang2014}
B.~Kang, V.~Monga, M.~Rangaswamy, Rank-constrained maximum likelihood
  estimation of structured covariance matrices, IEEE Transactions on Aerospace
  and Electronic Systems 50~(1) (2014) 501--515.
\newblock \href {https://doi.org/10.1109/TAES.2013.120389}
  {\path{doi:10.1109/TAES.2013.120389}}.

\bibitem{sun2016}
Y.~{Sun}, P.~{Babu}, D.~P. {Palomar}, Robust estimation of structured
  covariance matrix for heavy-tailed elliptical distributions, IEEE
  Transactions on Signal Processing 64~(14) (2016) 3576--3590.
\newblock \href {https://doi.org/10.1109/TSP.2016.2546222}
  {\path{doi:10.1109/TSP.2016.2546222}}.

\bibitem{razaviyayn2013}
M.~Razaviyayn, M.~Hong, Z.-Q. Luo, A unified convergence analysis of block
  successive minimization methods for nonsmooth optimization, SIAM Journal on
  Optimization 23~(2) (2013) 1126--1153.

\bibitem{royston2004}
P.~Royston, Multiple imputation of missing values, The Stata Journal 4~(3)
  (2004) 227--241.

\bibitem{bhatia2009positive}
R.~Bhatia, Positive definite matrices, Princeton university press, 2009.

\bibitem{bouchard2020}
F.~Bouchard, A.~Mian, J.~Zhou, S.~Said, G.~Ginolhac, Y.~Berthoumieu, Riemannian
  geometry for compound {G}aussian distributions: Application to recursive
  change detection, Signal Processing 176 (2020) 107716.
\newblock \href {https://doi.org/https://doi.org/10.1016/j.sigpro.2020.107716}
  {\path{doi:https://doi.org/10.1016/j.sigpro.2020.107716}}.

\bibitem{musial2011}
J.~P. Musial, M.~M. Verstraete, N.~Gobron, Technical note: Comparing the
  effectiveness of recent algorithms to fill and smooth incomplete and noisy
  time series, Atmospheric Chemistry and Physics 11~(15) (2011) 7905--7923.
\newblock \href {https://doi.org/10.5194/acp-11-7905-2011}
  {\path{doi:10.5194/acp-11-7905-2011}}.

\bibitem{zoubir2012}
A.~M. Zoubir, V.~Koivunen, Y.~Chakhchoukh, M.~Muma, Robust estimation in signal
  processing: {A} tutorial-style treatment of fundamental concepts, IEEE Signal
  Processing Magazine 29~(4) (2012) 61--80.

\bibitem{hf2020}
A.~Hippert-Ferrer, Y.~Yan, P.~Bolon, {EM-EOF}: Gap-filling in incomplete {SAR}
  displacement time series, IEEE Trans. Geosci. Remote Sens. 59~(7) (2021)
  5794--5811.

\bibitem{lerman2012robust}
G.~Lerman, M.~McCoy, J.~A. Tropp, T.~Zhang, Robust computation of linear
  models, or how to find a needle in a haystack, Tech. rep., California Inst of
  Tech Pasadena Dept of Computing and Mathematical Sciences (2012).

\bibitem{chen2011}
J.~Chen, X.~Zhu, J.~E. Vogelmann, F.~Gao, S.~Jin, A simple and effective method
  for filling gaps in {L}andsat {ETM}+ {SLC}-off images, Remote sensing of
  environment 115~(4) (2011) 1053--1064.

\bibitem{rakwatin2008}
P.~Rakwatin, W.~Takeuchi, Y.~Yasuoka, Restoration of {A}qua {MODIS} band 6
  using histogram matching and local least squares fitting, IEEE Transactions
  on Geoscience and Remote Sensing 47~(2) (2008) 613--627.

\bibitem{breizhcrops2020}
M.~Ru{\ss}wurm, C.~Pelletier, M.~Zollner, S.~Lef{\`e}vre, M.~K{\"o}rner,
  Breizhcrops: A time series dataset for crop type mapping, International
  Archives of the Photogrammetry, Remote Sensing and Spatial Information
  Sciences ISPRS (2020) (2020).

\bibitem{barachant2012}
A.~Barachant, S.~Bonnet, M.~Congedo, C.~Jutten, Multiclass brain–computer
  interface classification by {R}iemannian geometry, IEEE Transactions on
  Biomedical Engineering 59~(4) (2012) 920--928.
\newblock \href {https://doi.org/10.1109/TBME.2011.2172210}
  {\path{doi:10.1109/TBME.2011.2172210}}.

\bibitem{congedo2019}
M.~Congedo, P.~L.~C. Rodrigues, C.~Jutten, The {R}iemannian minimum distance to
  means field classifier, in: 8th Graz Brain-Computer Interface Conference
  2019, 2019.

\bibitem{vassilvitskii2006}
S.~Vassilvitskii, D.~Arthur, k-means++: The advantages of careful seeding, in:
  Proceedings of the eighteenth annual ACM-SIAM symposium on Discrete
  algorithms, 2006, pp. 1027--1035.

\bibitem{baumgardner2015}
M.~F. Baumgardner, L.~L. Biehl, D.~A. Landgrebe, 220 band aviris hyperspectral
  image data set: {J}une 12, 1992 indian pine test site 3, Purdue University
  Research Repository 10 (2015).

\end{thebibliography}

\end{document}